\documentclass{article}

\usepackage{fullpage}
\usepackage{amsmath,amssymb,amsthm,epsfig,graphicx, color}
\usepackage[linesnumbered, boxruled, lined]{algorithm2e}
\usepackage{latexsym, ae, caption}
\usepackage{psfrag}

\newcommand{\NP}{NP}
\newcommand{\POL}{P}

\newcommand{\N}{\mathbb{N}}

\def\fmax{f_{\mathrm{max}}}
\def\bmax{b_{\mathrm{max}}}
\def\falg{f_{\mathrm{alg}}}
\def\fopt{f_{\mathrm{opt}}}
\def\alg{\mathrm{ALG}}

\newtheorem{thm}{Theorem}

\title{Dynamic Defragmentation of Reconfigurable 
Devices\thanks{A  preliminary, considerably shorter extended-abstract version of this paper
appeared in the proceedings of the FPL~08 \cite{fktvakt-nbddr-08}.}}

\author{
S\'andor~P.~Fekete$^1$ 
\and Tom~Kamphans$^1$\footnote{Supported by
DFG grant FE~407/8-2, 8-3, project ``ReCoNodes'',
as part of the Priority Programme 1148, ``Reconfigurable Computing''.}
\and Nils~Schweer$^1$
\and Christopher Tessars$^1$\footnote{Supported by BMBF grant 03FEPAI2, project ``Advest''.}
\and Jan~C.~van~der~Veen$^1$\footnotemark[2]
\and Josef Angermeier$^2$\footnote{Supported by DFG grant TE~163/14-2, 14-3,
project ``ReCoNodes'', as part of the Priority Programme 1148,
``Reconfigurable Computing''.}
\and Dirk Koch$^2$\footnote{Supported by DFG
grant TE~163/13-2, 13-3, project ``ReCoNets'', as part of the Priority
Programme 1148, ``Reconfigurable Computing''.}
\and J\"urgen Teich$^2$}

\date{
  $^1$
  Department of Computer Science\\
  Braunschweig University of Technology\\
  Braunschweig, Germany\\
  email: \{s.fekete, t.kamphans, n.schweer,\\
  \ \ \ j.van-der-veen\}@tu-bs.de\\[15pt]
  $^2$
  Department of Computer Science 12\\
  University of Erlangen-Nuremberg\\
  Erlangen, Germany\\
  email: \{angermeier, dirk.koch, teich\}@cs.fau.de
}

\begin{document}

\maketitle

% Abstract.
% ---------
\begin{abstract}
We propose a new method for defragmenting the module layout of
a reconfigurable device, enabled by a novel approach for dealing
with communication needs between relocated modules and 
with inhomogeneities found in commonly used FPGAs. 
Our method is based on dynamic
relocation of module positions during runtime, with only
very little reconfiguration overhead; the objective is to
maximize the length of contiguous free space that is available
for new modules. We describe a number of algorithmic
aspects of good defragmentation, and present an optimization
method based on tabu search. Experimental results indicate that
we can improve the quality of module layout by roughly 50\% over
static layout. Among other benefits, this improvement
avoids unnecessary rejections of modules.
\end{abstract}

% First section, often named as Introduction.
% -------------------------------------------
\section{Introduction}
\label{sec:intro}
\subsection{Reconfiguration and Communication}
FPGAs combine the performance of an ASIC implementation with the flexibility of software realizations. 
Partial runtime reconfiguration is an applicable technique to overcome significant area overhead,
monetary cost, higher power consumption, or speed penalties as compared to ASICs (see e.g.~\cite{rose_FPGAgap}).
By loading just the required modules to an FPGA at runtime, it is possible to
build smaller systems and less power-hungry devices.
For instance, an embedded system may start up with some boot-loader and test
modules. These modules may be exchanged by a crypto-accelerator to
speed up the authentication process of the user.
Later, different modules will be loaded to the FPGA by partial runtime
reconfiguration with respect to the user demand or the state of the system.
Note that many systems provide mutually exclusive functionality
(e.g., the record or the play mode of a multimedia device) that is suitable
to share some FPGA resources at runtime.
Furthermore, modules need to communicate with other modules to accomplish their tasks. Therefore,
a suitable communication infrastructure must be applied and the
implied costs in terms of time and area resources must be respected. This
challenge and possible solutions are discussed in Section~2.

\begin{figure}[h!]
\centering
%\vspace{-5 mm}
\includegraphics[width=.8\linewidth]{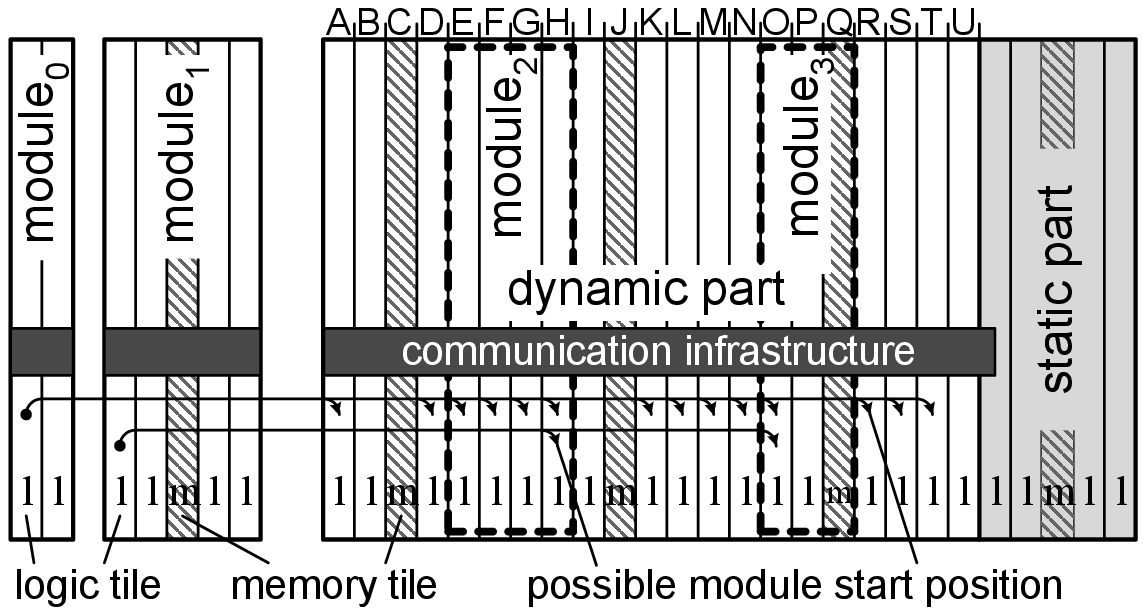}
\caption{\label{fig:fpga_system}Dynamically reconfigurable, tile-oriented system.
The system shares some logic tiles $l$ and memory tiles $m$ among
a set of modules within the dynamic part of the system.
Some modules require a memory tile at a fixed offset with respect
to the start position within the modules
(e.g., the third tile of \textsf{module}$_\textsf{1}$ is a memory tile).
}
\end{figure}

When using such systems, an efficient resource management becomes necessary.
One problem that has to be solved at runtime is the fragmentation of the
tiles due to the time-dependent execution of some modules
on the same resource area.
It is assumed that for dynamically partially reconfigurable systems,
modules are to be vertically aligned column by column, as shown in
Fig.~\ref{fig:fpga_system}.
Accordingly, a module requiring multiple tiles to implement its logic 
will demand a consecutive adjacent set of tiles without any gaps.
This problem is discussed in this paper.

\subsection{Dynamic Storage Allocation on Reconfigurable Devices}
%DK The ever-increasing capabilities of modern reconfigurable devices
The ever-increasing capabilities of modern reconfigu\-rable devices
give rise to a large number of new challenges; solving one of them
in turn gives rise to new possibilities and challenges.
As described above, there are new solutions for dealing with
the communication of relocated devices; this opens up new
possibilities for dynamic relocation of modules.
The resulting challenge is the dynamic
allocation of module requests to a reconfigurable device:
%DK
given an array-shaped reconfigurable device and a sequence of module requests of
%DKvarying size, assign each module to a contiguous set of slots
varying resource requirements (e.g., logic tiles or memory blocks), 
assign each module to a contiguous set of slots
on the device; see Fig.~\ref{storagealloc:fig}(a).

\begin{figure}[t]
  \centerline{\epsfig{figure=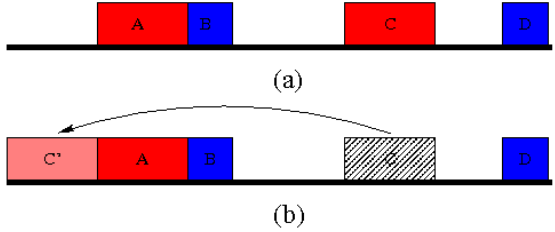,width=74mm}}

\caption{Dynamic storage allocation:
(a) Each module occupies a contiguous block of array positions.
(b) Moving a module to a new position in order to increase maximum
free interval size.}\label{storagealloc:fig}

\end{figure}

At first glance, this problem has a striking resemblance to
one of the classical problems of computing:
{\em Dynamic storage allocation} considers
a memory array and a sequence of storage requests of varying
size, looking for an assignment of each request to a 
contiguous\footnote{Note that this part of the comparison 
refers to classical research;
of course modern storage devices place
%DK virtual memory blocks on discontiguous physical space, at the 
virtual memory blocks on discontiguous physical space, at the 
expense of extra overhead for the pointer structures. This
%CK approach for allocating dicontiguous
approach for allocating discontiguous
space is not possible for placing the modules on a reconfigurable 
device, which is the challenge faced by this paper.}
block of memory cells, such
that the length of each block corresponds to the size of the request.
Once this allocation has been performed,
it is static in space: after a block has been occupied,
it will remain fixed until the corresponding data is no longer needed
and the block is released. As a consequence, a sequence of
allocations and releases can result in fragmentation of
the memory array, making it hard or even impossible to store
new data.

Over the years, a large variety of methods and results for
allocating storage have been proposed. The classical sequential fit
algorithms, First Fit, Best Fit, Next Fit and Worst Fit can be found
in Knuth~\cite{Knuth97} and Wilson et al.~\cite{Wils95}.

Buddy systems partition the storage into a number of standard block
sizes and allocate a block in a free interval of the smallest
standard size sufficient to contain the block. Differing only in the
choice of the standard size, various buddy systems have been 
proposed
\cite{Bromley80,Hinds75,Hirs73,Know65,Shen74,Knuth97}.
Newer approaches that use cache-oblivious structures
for allocating space in memory hierarchies
include the works by Bender et al.~\cite{Bender05,Bender05a}.

There are notable differences between the dynamic allocation of
modules to a reconfigurable device and dynamic storage
allocation. First of all, all modules on a reconfigurable device may execute
in parallel, while on a standalone processor, large blocks in memory
are not used simultaneously. Reconfigurable devices do not provide 
techniques such as
paging and virtual memory mapping that allow arranging memory blocks
next to each other in a virtual way, while they are physically stored 
at non-adjacent
positions. The reconfiguration of a module on a reconfigurable device
implies delays, and an inter-module communication infrastructure is
required, because the functionality of a reconfigurable device may depend on
other modules and external periphery.

Modules on a reconfigurable device can be relocated to a different location on
the reconfigurable device, this can even be done at runtime.  However, today's
synthesis tools still lack support for placing a module
implementation at different positions: these tools often allow placing a
module at only one specific position; thus, we cannot use the same
implementation binary for different positions on the reconfigurable device.
Different techniques have been conceived to tackle this problem. One solution
is to equip the reconfigurable device with a special reconfiguration-management
unit that handles the modification of the module implementations at runtime
such that they can be placed at the desired position. Moreover, in order to
relocate a running module, the module must be paused, the state must be
temporarily saved, the module must be reconfigured at the new position, the
state must be restored, and the module must get a signal to continue its work.
Different techniques have been developed for this task, one of them is
presented by Koch et al.~\cite{KHT07}. In the future, reconfigurable devices
may have additional support for task preemption.

In contrast to memory and storage devices, reconfigurable devices often contain
heterogeneities such as dedicated memories, DSPs, or CPUs. These units enable
or increase performance in important application fields. But heterogeneities
increase the complexity of defragmentation considerably: a module
implementation possibly depends on a specific pattern of heterogeneous
resources at the placement location in order to complete its task. The number
of feasible positions of a module on a FPGA can be increased by creating
different implementations of the same module (i.e., with different positions
for the heterogeneities), but this approach also requires additional storage
space for module implementations. Having different implementations of a module
also increases the number of possibilities when defragmenting the module
placements. Thus, the complexity of the defragmentation problem increases. 

There is a huge amount of related work also from within the FPGA
community: Becker et al.~\cite{blc-erpbr-07} present a method for
enhancing the relocatability of partial reconfigurability of partial
bitstreams for FPGA runtime configuration, with a special focus on
heterogeneities. They study the underlying prerequisites and
technical conditions for dynamic relocation. In the process, a method that
circumvents the problem of having to find fully identical regions for
the modules is solved by the creation of compatible subsets of
resources, enabling a flexible placement of relocatable modules.
Gericota et al.~\cite{gericota05} present a relocation procedure for
Configurable Logic Blocks (CLBs) that is able to carry out online
rearrangements, defragmenting the available FPGA resources without
disturbing functions currently running. Another relevant approach was
given by Compton et al.~\cite{clckh-crdrt-02}, who present a new
reconfigurable architecture design extension based on the ideas of
relocation and defragmentation. It is shown that with
little run-time effort on the part of the CPU and little additional
area-increase over a basic partially reconfigurable FPGA, the
reconfiguration overhead can be reduced tremendously. Koch et
al.~\cite{kabk-faepm-04} introduce efficient hardware extensions to
typical FPGA architectures in order to allow hardware task
preemption. Furthermore, the technical aspects of applying 
hardware task preemption to avoid defragmentation are discussed.
These papers do not consider the algorithmic implications and how the
relocation capabilities can be exploited
to optimize module layout in a fast, practical fashion, which is what we consider in this paper. Koester et
al.~\cite{koester07} also address the problem of
defragmentation. Different defragmentation algorithms that minimize
different types of costs are analyzed. With the help of a simulation
model and a benchmark, simulation results and algorithm comparisons
are presented. However, the problem description differs in some major points;
for example, no heterogeneities in the reconfigurable area are considered.

The general concept of defragmentation is well known, and has been applied to
many fields, e.g., it is typically employed for memory management. Our approach
is significantly different from defragmentation techniques which have been
conceived so far: these require a freeze of the system, followed by
a computation of the new layout and a complete reconfiguration of all modules
at once.  Instead, we just copy one module 
at a time, and simply switch the execution to the new module as soon as the move is complete.
This leads to a {\em seamless, dynamic defragmentation
of the module layout}, resulting in much better
utilization of the available space for modules.

The rest of this paper is organized as follows. In the following Section~2
we give a description of the underlying model and assumptions of the reconfigurable device and application, giving rise
to the problem description in Section~3. As it turns out, solving
the corresponding optimization problem is NP-hard, as shown
in Section~4. 
However, for moderate module
density, it is still possible to compute optimal results,
as shown in Section~5. 
In Section~6, we show that there
are instances for which $\Omega(n^2)$ moves are necessary.
This leads to a heuristic optimization method 
for higher densities, based on tabu search and described
in Section~7. Detailed experimental results are presented and discussed in Section~8 showing an increase in the maximal free space in average by $25\%$ when applying our defragmentation techniques for FPGAs with heterogeneities. On some inputs an increase up to $200\%$ is observed. 
Concluding thoughts are presented in Section~9.

% Second section. % --------------- 

\section{Problem Scenario and Technical Challenges} 
\label{sec:scenario}

% assumption: reconfig overhead neglectable
Each partial reconfiguration of a module on a reconfigurable
device incurs a certain amount of reconfiguration overhead. The
ratio between the reconfiguration time and the actual running time of
the corresponding modules is highly application specific. We assume in
our scenario that the reconfiguration time is sufficiently small
compared to the execution times of modules used. Of course,
there are applications in which the reconfiguration overheads
must be taken into account, because many different modules are loaded on the
reconfigurable device and their execution times are not much higher than their
reconfiguration times. However, the possibility of reconfiguring only a part of
the reconfigurable device as well as techniques such as prefetching, latency hiding,
and bitstream compression can significantly reduce the reconfiguration
overheads. Furthermore, even today, for many applications a module's reconfiguration
time is much less than its execution time even today. So far, it is not known
whether reconfiguration overheads will still play an important role for the
performance of many applications in the future or not. In this paper, we assume
that there will be also many applications in the future for which the
reconfiguration overheads are no big issue.

%%%%%% Weg wegen Reviewer-1-neu
%%%%%% and the portion of these application is increasing over time, as the reconfiguration times drop with each new generation of reconfigurable devices tremendously. See Table \ref{tab:reconfigs} for the theoretical and practically achieved throughput of the reconfiguration interface on well-known reconfigurable devices (for further details see also Claus et al.~\cite{claus09reconfig}). 
%%%%%%
%%%%%%
%%%%%%
%%%%%%\begin{table}
%%%%%%\begin{center}
%%%%%%  \begin{tabular}{| c | c | c | c | c |}
%%%%%%    \hline
%%%%%%    Device type & Interconnect & Frequency & Theo. thr. [MB/s] & Pract. thr. [MB/s]\\ \hline \hline
%%%%%%    Virtex-II & ethernet & 100 & 100 & 50 \\ \hline
%%%%%%    Virtex4 & OPB & 100 & 400 & 350 \\ \hline
%%%%%%    Virtex4 & PLB & 100 & 400 & 371.4 \\ \hline
%%%%%%    Virtex5 & MPMC & 200 & 400 & 400 \\ \hline
%%%%%%  \end{tabular}
%%%%%%\end{center}
%%%%%%\caption{Throughput of ICAP reconfiguration controller on XUPV2P, ML405 and ML507 board (see also Claus et al.~\cite{claus09reconfig}).}
%%%%%%  \label{tab:reconfigs}
%%%%%%\end{table}

% communication overheads
%DK \cite{hagemeyer_ersa07} introduces a system that provides a bus-based
In order to take more benefit from runtime reconfiguration, 
systems should be able to provide the reconfigurable resources in a very
flexible way to the modules. Therefore, a communication infrastructure is
required, such that modules can communicate with each other, and to peripheral
input/output devices.  Most related work for reconfigurable communication
systems is still based on the assumption that the locations allowed for modules
in a partially reconfigurable system are all fixed in size (e.g., Lysaght et
al.~\cite{xilinx_fpl06}). Consequently, such approaches do not allow for
exchanging a large module with multiple smaller ones. This originates from a
lack of adequate communication techniques suitable to connect multiple
partially reconfigurable modules within the same resource area to the rest of
the system.
However, there are notable exceptions: Koch et al.~\cite{koch_fpl08,koch_fccm} present a system with a reconfigurable area
partitioned into 60 tiles, each capable of connecting 
a tiny 8-bit module to the system using
the so-called ReCoBus.
This allows it to implement larger interfaces or modules by combining
multiple
adjacent tiles; e.g., 4~tiles are required for building a 32-bit
interface.
In addition, the ReCoBus can link I/O pins to the partially reconfigurable
modules. Furthermore, this approach of a reconfigurable bus demonstrates
that high placement flexibility, low resource overhead, and high throughput can
be achieved at the same time.

In some partially reconfigurable computing systems, module communication in a
neighbor-to-neighbor-based manner is preferred to using a
reconfigurable bus system:  for example, FPGAs are also used in streaming
applications, such as video processing and packet processing, where each module
communicates concurrently with the next module in the pipeline such that the
communication costs are kept low. If these systems will also benefit from
defragmentation techniques highly depends on the communication constraints of
the modules and on the individual reconfigurable computing system. In general,
one option is to change the communication infrastructure of the modules to a
more flexible system, such as a reconfigurable bus system. This may lead to
increased communication costs, but at the same time, defragmentation techniques
can place modules more freely, and thus yield better results. If the increase
in communication costs is clearly amortized by the improvements due to better
defragmentation results, then the reconfigurable system will benefit from this
option. In a setting in which some modules in a reconfigurable computing system
must be placed closely to each other (e.g., they may strongly rely on a fast
neighbor-to-neighbor communication for performance reasons),
these modules can be grouped together such that they are considered by the
defragmentation strategy as a single module. Therefore, either all these
modules are moved to another position for defragmentation, or no module is
touched. Thus, defragmentation techniques for reconfigurable devices are
flexible enough to accommodate all important technical aspects concerning
module communication on FPGAs. 

% all kinds of heterogeneities
So far, there already exists an enormous and ever-growing number of different
reconfigurable devices. Most of their
reconfigurable area consists of heterogeneities, special-purpose
units such as DSPs, CPUs, or RAMs, which offer a considerable performance
improvement for target applications. 
See, for example, Fig.~\ref{fig:fpga_system}: this FPGA has two
different column types,
\emph{logic tiles} ($l$) and \emph{memory tiles} ($m$).
The important challenge with heterogeneities are placement
limitations: modules applying special-purpose units may not be freely
relocated, but can be placed only at positions offering the same
geometry of special purpose units;
the placement of a module within the reconfigurable resource area on 
the FPGA must fit exactly to the particular module. 
%% The requirement of a module can be modeled as a string and the search 
%% of a valid placement position is then a string matching problem 
%% in the reconfigurable resources provided by the dynamic part of the system. 
%% This restricts the possible module start positions, and 
Thus, the number of free tiles is 
not sufficient to determine whether a module can be placed. 
For instance, \textsf{module}$_\textsf{1}$ in Fig.~\ref{fig:fpga_system} 
has the resource requirement $l\,\,l\,m\,l\,\,l$ and can be placed only at 
the positions \textsf{A}, \textsf{H}, and \textsf{O}, which are currently occupied by 
\textsf{module}$_\textsf{2}$ and \textsf{module}$_\textsf{3}$. 
In the example, the system has 12 free logic tiles and 2 free memory tiles, 
but we are currently not able to place \textsf{module}$_\textsf{1}$ 
on the FPGA, which requires just 4 logic tiles and 1 memory tile. 
Note that our approach does not depend on a
specific type of heterogeneity, it can also be applied to future
reconfigurable devices with new kinds of heterogeneities.

% only column-wise reconfiguration
Our approach is targeted at currently available FPGAs and
future reconfigurable devices. In our problem formulation, we
assume a device that is capable of column-wise partial
reconfiguration, i.e., only whole columns of the reconfigurable area are
exchanged. Modern reconfigurable devices offer also the flexibility to
reconfigure single cells in the reconfigurable area, but this kind of
higher flexibility is not assumed in our problem formulation, because the
column-wise reconfiguration is considered as an important case for
these studies. Therefore, one reason may be that the applied device
cannot provide that kind of higher flexibility, e.g., in order to
save unnecessary costs. Many applications for reconfigurable devices
work in a pipeline-based manner and employ modules that span over the
whole column. They use only modules with the same heights, because allowing 
a greater level of flexibility concerning the placement would also
imply higher resource overheads, e.g., in terms of communication
resources. Furthermore, as long as the heights of all modules are
equal, our approach can also be applied to cell-based reconfigurable
devices using new abstraction layer: we introduce a new type of heterogeneity 
(the ``'separating heterogeneity'') that is not used by any module. Then we
simply connect the horizontal lines of cells of the device to form a single
row, separated by this new heterogeneity; see Fig.~\ref{fig:2Dto1D}. Thus, any
placement of a module on the abstract device can be mapped to a placement on
the original device.

\begin{figure*}[ht]
\begin{center}
 \input{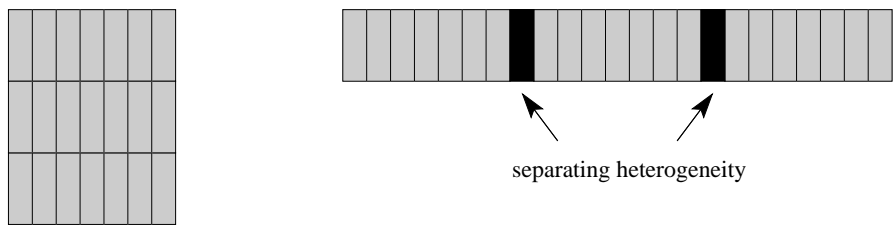}
\end{center}
\caption{\label{fig:2Dto1D}Defragmenting a two-dimensional device. 
Left: A two-dimensional device. Right: The corresponding one-dimensional
device with separating heterogeneities.}
\end{figure*}

Our studies of the important case of column-based
reconfiguration can also be applied to scenarios in which a
cell-based reconfiguration and modules with differing heights
are needed: the local search techniques applied in our approaches can
also be used for finding another suitable place for a module in the
two-dimensional space. The decision which steps to choose can also be
extended from one to two dimensions. Thus,
the proposed approach is not strictly limited to the important case of
column-wise reconfiguration.

When modules are relocated for defragmentation, we have to distinguish between
moving only the module configuration and the configuration together with
the internal state. 
In the first case, we just make a copy of the reconfiguration data to the 
new position and start the next computation on the module at the new 
position (e.g., a discrete cosine transformation on the next 
frame in a video system). 
In the second case, both modules have to be interrupted and the state 
(represented by all internal flip-flop and memory values) will be copied
to the target module. 
It may not be enough to copy the configuration data to a new position,
because the configuration bit files often imply a certain position. Therefore,
it is either necessary to alter dynamically the bit files, or to generate
statically bit files for all possible positions. Relocation of modules and
related problems were already addressed in other works.  Furthermore, the
communication between modules must be stalled during the relocation of the
respective modules. Thus, the communication infrastructure should be flexible
enough to meet these requirements. 
As compared to the reconfiguration process, copying the state can be performed
with short interruption when using hardware checkpointing (for more details see
Koch et al.~\cite{KHT07}). 
 
If we allow overlapping regions for the defragmentation, e.g., the source and
the target module may overlap, then the interruption time can be dominated by
the relocation process: an overlap prevents the possibility to copy the routing
information and logic settings to the destination, while the original module is
still running. In this case, the module must be stopped, the reconfiguration
data and the state of the module must be copied to some (external) memory, and
be restored at the destination. This procedure takes longer if the
regions overlap. As a consequence, we will prevent our defragmentation
algorithms from using overlapping regions to place modules. Thus, switching from
the original module to the new one can be optimized in such a way, that no
input data is lost, and the downtime of the module is minimized. Thus, a
copy of module---without the state---can be reconfigured at the destination
while the module is still running. Therefore, switching between the two modules
is very fast for modules that have only few state data to be copied.
Furthermore, our proposed defragmentation strategies move at most one module at a time to
another position on the reconfigurable area. Thus, only a single module is
affected at a moment by the defragmentation process, while the remaining set of
modules remains untouched. 

\section{Problem Formulation}
\label{sec:problem}

In this paper, we consider a reconfigurable device
that allows allocating modules in a contiguous manner
on an array $L$ of length $\ell$;
modules will be denoted by $M_1,\ldots,M_n$. A module $M_i$
placed in the array occupies a contiguous interval on the reconfigurable device,
denoted by $L_{M_i}$. Modules are always placed such that $L_{M_i}
\cap L_{M_j} = \emptyset$ for $i\neq j$; that is, two different modules
do not overlap.

Modules placed in the array divide $L$ into sections
that are occupied by a module and sections that are not occupied;
the latter are called {\em free intervals}, denoted by $F_1, \ldots, F_k$.
% bisher nicht benoetigt:----------------------------------------------
%For a module $M_i$ we define $F^{\rght}_{M_i}$ to be the free
%interval to the right of $M_i$. If $L_{M_i}$ shares the rightmost
%point with a interval $L_{M_j}$, we define $F^{\rght}_{M_i}$ to
%be of size $f^{\rght}_{M_i}=0$. $F^{\lft}_{M_i}$ is defined in a
%similar manner.
%----------------------------------------------------------------------
Partially reconfigurable devices allow us to relocate a module $M_i$ of size
$m_i$
from interval $L_{M_i}$ to a new position $L'_{M_i}$ within
a free interval $F_j$ of size $f_j$,
provided that the following two conditions are fulfilled:
\begin{enumerate}
\item[$\bullet$] 
No occupied section is chosen (i.e., 
$L'_{M_i} \cap \bigcup_{k=1}^n L_{M_k} = \emptyset$).

\item[$\bullet$] 
The size of the free interval is at least as big as the size of the module:
(i.e., $f_j \geq m_i$).

\end{enumerate}

% Note that the first condition implies that $L_M$ and $L_{M_i}$ are
% not allowed to intersect; that is, the interval occupied by $M_i$
% before the move and the new interval have to be disjoint.
% Furthermore, it ensures that the new position is not occupied by any
% other module. The second condition ensures that there is sufficient
% free space to provide a new position for the module.

The {\em Maximum Defragmentation Problem} (MDP) asks for a
sequence of relocation moves that maximizes the size of the largest
free interval on the reconfigurable device.  We distinguish between
the {\em homogeneous} MDP, in which every cell in the array is
equivalent, and the {\em heterogeneous} MDP, which accounts for
heterogeneities in the given FPGA. Clearly, the heterogeneous MDP is
more difficult. Thus, we focus on the homogeneous MDP for our
complexity results, as their harndess implies hardness of the more complicated,
restricted versions.

\begin{figure*}[t]
\begin{center}
 \input{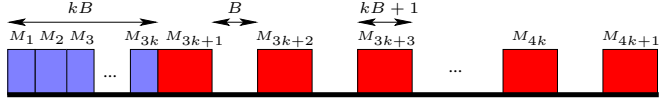}
\end{center}
\caption{\label{fig:npcomlete}Reducing 3-Partition to the MDP.}
\end{figure*}

The larger free interval after the defragmentation can allow to place and execute a module
that could not be placed before.  Moreover,
defragmentation helps to place modules at an earlier time.
Altogether, the {\em makespan} is reduced, i.e., the total
time that is needed to satisfy a sequence of requests (i.e., a
sequence of modules $M_1, \ldots, M_n$), considering that
every module $M_i$ needs a certain time, the {\em duration} $T_i$, to
run on the FPGA before it can be removed.

\section{Problem Complexity}
\label{sec:complexity}

In this section, we state two complexity results for defragmenting modules on a reconfigurable device: one for deciding whether
one contiguous free block can be formed,
and one for the maximization version of the (homogeneous) defragmentation
problem. We show that the decision version is strongly \NP-complete
and that no approximation algorithm with a
useful approximation factor exists for the maximization version,
unless \POL=\NP.

We use a proof technique know as {\em proof by reduction}. That is, we 
take a problem that is known to be hard and show how to transform 
an instance of the known problem to an instance of our problem. 
Thus, if we had an efficient method for solving our problem, it could also be
used for solving the other, hard problem.
The problem {\sc 3-Partition} is the main ingredient of the reduction. It
belongs to the class of strongly NP-complete and can be stated as
follows \cite{Garey79}:

\begin{description}
\item[Given:] A finite set of $3\cdot k$ elements $C_1,\ldots,C_{3k}$ with sizes
$c_1,\ldots,c_{3k}$, a bound $B\in \N$ such that $c_i$ satisfies
$\frac{B}{4} < c_i < \frac{B}{2}$ for $i=1,\ldots,3k$ and 
$\sum_{i=1}^{3k} c_i = k\cdot B$.

\item[Question:]  Can the elements be partitioned into $k$ disjoint sets
$S_1,S_2\ldots,S_k$, such that for $1\leq \ell \leq k$, $\sum_{C_j \in
S_\ell} c_j = B$?
\end{description}

Because the $c_i$ have a lower bound of $\frac{B}{4}$ and an upper
bound of $\frac{B}{2}$, each set $S_j$ contains exactly three
elements. We state our complexity result:

\begin{thm}
The Maximum Defragmentation Problem with free intervals $F_1,\ldots,
F_k$ is strongly \NP-complete.
%; moreover, the problem does not allow
%any deterministic polynomial-time approximation algorithm within any
%polynomial approximation factor, unless \POL=\NP.
\end{thm}

\begin{proof}
Given an instance of the problem {\sc 3-Partition} with input
$c_1,\ldots,c_{3k}$ and bound $B$, we construct an instance of the
MDP in the following way: We place $3k$
modules $M_1,\ldots,$ $M_{3k}$ with $m_i=c_i$, $1\leq i \leq 3k$, side
by side, starting at the left end of $L$. Then starting at the right
boundary of $M_{3k}$, we place $k+1$ modules of size
$kB+1$, alternating with $k$ free intervals of size $B$. We denote 
these modules by
$M_{3k+1}$ to $M_{4k+1}$ and the free intervals by $F_1$ to $F_k$. 
Fig.~\ref{fig:npcomlete} shows the overall structure of the constructed
instance. Now we ask for the construction of a free interval of size
$K=k\cdot B$. Because the size of the total free space is equal to $kB$,
none of the modules $M_{3k+1},\ldots,M_{4k}$ can ever be moved.
Hence, the only way to connect the total space is to move the
modules $M_1$ to $M_{3k}$ to the free intervals. But any solution of
this kind implies a solution to the given instance of {\sc 3-Partition}, concluding the proof of NP-completeness.
% 
% This construction can also be used to establish the claimed
% inapproximability result: as can be seen from Fig.~\ref{fig:npcomlete}~(b),
% the gap between the possible solution values in case of existence
% and nonexistence of a 3-Partition can be made arbitrarily big.
\end{proof}

Proving \NP-completeness for the decision version of a problem 
makes it interesting to consider {\em approximating} the size of the maximal
constructable free intervals: instead of finding the best possible value
$\fopt$, we may be content with an approximate value $\falg$, as long as it
can be found in polynomial time and is within a constant factor of $\fopt$.
The next theorem
shows that the existence of any algorithm with a useful approximation
factor is unlikely, even if we only require an asymptotic factor.

\begin{thm}
\label{th:2}
Let $\alg$ be a polynomial-time algorithm with $\fopt \leq \alpha
\cdot \falg + \beta$. Unless \POL=\NP, $\alpha$ must be big, i.e.,
$\alpha \in \Omega( (n\cdot\max\{\log\fmax,\log\bmax\})^{1-\varepsilon} )$, for any
$\varepsilon > 0$, 
where $n$ denotes the number of modules, $\fmax$ denotes the size of the largest free interval in the input, 
and $\bmax$ the size of the largest module.
%% There is no approximation algorithm for the maximum defragmentation
%% problem with an approximation factor polynomially bounded in the
%% input size of the problem, unless \POL=\NP.
\end{thm}

\begin{proof}
Refer to Fig.~\ref{fig:nolinapx}.
We will show that if $\alg$ is an $\alpha$-approximation algorithm for \linebreak
$\alpha \in O( (n\cdot\max\{\log\fmax,\log\bmax\})^{1-\varepsilon} )$, it
can be used to decide whether a {\sc 3-Partition} instance is solvable. For a given
instance with numbers $c_1,\ldots,c_{3k}$ and a bound
$B\in \N$ (recall that $\frac{B}{4} < c_i < \frac{B}{2}$), we
construct an allocation of modules inside an array, as shown in the figure.
Starting at the left end of the array we
place $3k$ modules side by side with $b_i = c_i$, for
$i=1,\ldots,3k$. Then, starting at the right boundary of $M_{3k}$,
we place $k+1$ modules of size $N = kB + 1 + rB/2$ (where $r$ is an arbitrary number of sufficient polynomially bounded size; more details will follow), 
alternating with
$k$ free spaces of size $B$. Now, for $i=1,\ldots,r$, we proceed
with a free space of size $B/4$, a module of size $b_{4k+2i} = kB +
(i-1)B/2$, a free space of size $B/4$, and a module of size
$b_{4k+2i+1} = N$.

\begin{figure}[t]
\psfrag{N}{$N$} %
\psfrag{Be}{$B$} %
\psfrag{B1}{$M_1$} %
\psfrag{Bi}{$M_i$} %
\psfrag{B3k}{$M_{3k}$} %
\psfrag{B3kp1}{$M_{3k+1}$} %
\psfrag{B3kp2}{$M_{3k+2}$} %
\psfrag{B3kp3}{$M_{3k+3}$} %
\psfrag{B4k}{$M_{4k}$} %
\psfrag{B4kp1}{$M_{4k+1}$} %
\psfrag{B4kp2}{$M_{4k+2}$} %
\psfrag{B4kp3}{$M_{4k+3}$} %
\psfrag{B4kp2im3}{$M_{4k+2i-3}$} %
\psfrag{B4kp2im2}{$M_{4k+2(i-1)}$} %
\psfrag{B4kp2im1}{$M_{4k+2i-1}$} %
\psfrag{B4kp2i}{$M_{4k+2i}$} %
\psfrag{B4kp2ip1}{$M_{4k+2i+1}$} %
\psfrag{Bd4}{$\frac{B}{4}$} %
\psfrag{kB}{$kB$} %
\psfrag{kB1}{$kB+\frac{i-2}{2}B$} %
\psfrag{kB2}{$kB+\frac{i-2}{2}B + 2\cdot\frac{B}{4}$} %
\psfrag{c1}{$c_1$} %
\psfrag{ci}{$c_i$} %
\psfrag{c3k}{$c_{3k}$} %
\centering \epsfig{file=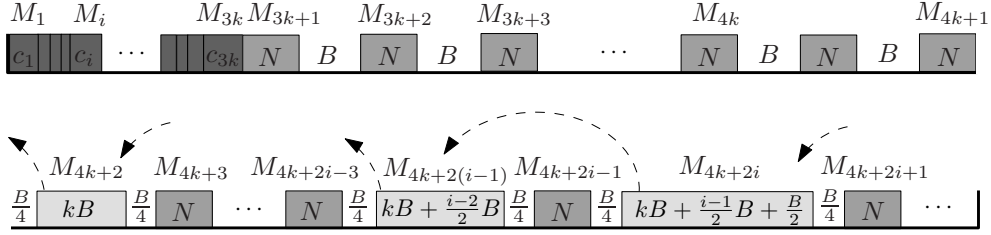,
width=13cm} \caption{\label{fig:nolinapx} The basic idea for the proof of Theorem~\ref{th:2}.
Modules are not drawn to scale; in particular,
modules of size $N=kB +
1 + rB/2$ (gray) for an appropriately large integer $r$ are very large, so they can never be moved. 
The modules $M_1,\ldots,M_{3k}$ encode an instance of {\sc 3-Partition}.
Module $M_{4k+2}$ of size $k\cdot B$ can be
moved if and only if these first $3k$ modules (dark-gray) can be moved to the first
$k$ free spaces of size $B$. Only in this case, for $i=2,\ldots,r$, the modules
$M_{4k+2i}$ (light-gray) fit exactly between $M_{4k+2i-3}$ and
$M_{4k+2i-1}$, increasing the size of the largest free interval by
$B/2$ with every move.}
\end{figure}

Note that the number of modules is $n= 5k + 4r +1$ and $\max\{\log\fmax,\log\bmax\} =
\log\bmax$. We claim that $\falg \geq kB$ if and only if the answer
to the {\sc 3-Partition} instance is ``yes''.

If $\falg \geq kB$, consider the situation in which the first free
space of size $kB$ occurs. Because none of the modules
$M_{3k+1},\ldots,M_{4k+2r+1}$ could be moved so far, and because
the modules $M_{1},\ldots,M_{3k}$ are larger than $B/4$, the only
way to create a free space of size $kB$ is to place the first $3k$
modules in the $k$ free spaces of size $B$. This implies a solution to
the {\sc 3-Partition} instance.

If $\falg < kB$, we show that the 
instance of {\sc 3-Partition} 
cannot be solved. If $\falg < kB$, then
$$
\fopt \leq \alpha \cdot \falg + \beta <
% C \cdot ( n\log \bmax )^{1-\varepsilon} \cdot kB + \beta =
kB \cdot C \cdot ( n \log \bmax )^{1-\varepsilon} + \beta \; ,
$$
for some constant $C$. The total free space has size $f = kB +
rB/2$. Because $n =5k + 4r +1$, $\bmax = N = kB + 1 + rB/2$ and
$k$, $B$, $C$, and $\beta$ are constant a straightforward
computation shows that
$$
\fopt < kB \cdot C \cdot \left[ (5k+4r+1) \log (kB + 1 +
\frac{rB}{2}) \right]^{1-\varepsilon} \hspace{-0.5cm} + \beta \leq
kB + \frac{rB}{2} = f \; ,
$$
for large $r$ (i.e., we choose $r$ such that the second inequality
holds).
Hence, a free space of size $kB + rB/2$ cannot be
constructed.

Conversely, a solution to the {\sc 3-Partition} instance allows the construction of a
free space of size $kB + rB/2$ as follows. The first $3k$ modules
are moved to the $k$ free spaces of size $B$. Now, $M_{4k+2}$ is
moved to the free space of size $kB$ and then, one after the other,
$M_{4k+2i}$ is moved between the modules $M_{4k+2i-3}$ and
$M_{4k+2i-1}$, for $i=2,\ldots,r$.

Thus, we can conclude that the existence of a polynomial-time approximation 
method for the MDP can be used to decide the feasibility of {\sc 3-Partition} instances,
i.e., implies \POL=\NP.
\end{proof}

\section{Moderate Density}
\label{sec:dense} 

The number of modules, $n$, their sizes, and the amount of free space on
the reconfigurable area are highly dependent on the application to be
executed and, furthermore, may also vary enormously during the execution of the application on a reconfigurable device. Initially or
at some later point in time, only a portion of the reconfigurable
device may be used. For a rather moderate density, we
conceived an efficient defragmentation routine.
We consider a special case in
which the homogeneous MDP can be solved with linear computing time in at most
$2n$ moves. We define the density of an array, $L$, of length $\ell$ 
to be $\delta := \frac{1}{\ell}\sum_{i=1}^n m_i$. We show that if
the density is bound by 
$\frac12(1-$ the fraction of the total area occupied by the largest module),
i.e.,
\begin{equation}
\delta \leq \frac{1}{2} - \frac{1}{2} \cdot\frac{\max_{i=1,\ldots,n}\{m_i\}}
{\ell}
\label{eq:lowden2}
\end{equation}
the total free space can always be connected with $2n$ steps by 
Algorithm~1.
The idea of the Algorithm~1 is to start with the leftmost module,
and shift all modules as far as possible to the left,
one after the other. In the second loop, we start with the 
rightmost module and shift modules as far as possible
to the right, one after the other. As it turns out, this results in one
connected free space. (Note that in some cases, a single round of shifts
is sufficient, which can easily be detected; however, two rounds may be necessary
if the initial configuration has small free intervals on the left.)

\begin{algorithm}[h!]\caption{LeftRightShift}
\label{alg:leftrightshift}

\SetLine

\KwIn{A array $L$ with $n$ modules $M_1,\ldots, M_n$ such that
(\ref{eq:lowden2}) is fulfilled.}

\KwOut{A placement of $M_1,\ldots, M_n$ such that there is only one
free interval at the left end of $L$. }

%\SetVline

\For{  $i=1$ \KwTo $n$ } {
        %\uIf{$ f_{M_i}^{left} \geq m_i$ } {
        Shift $M_i$ to the left as far as possible.} % }
\For{  $i=n$ \KwTo $1$ } {
        Shift $M_i$ to the right as far as possible.}

\end{algorithm}

For proving correctness 
of Algorithm \ref{alg:leftrightshift}, 
we need the following two observations.
Both follow immediately from
the definition of density and from (\ref{eq:lowden2}); in the following,
$f_i$ denotes the size of free intervals $F_i$.

\begin{equation}
\sum_{i=1}^k f_i \geq \frac{l}{2} + \frac{1}{2}\cdot
\max_{i=1,\ldots,n}\{m_i\} \label{eq:lowden3}
\end{equation}
\begin{equation}
\delta < \frac{1}{2} \;\; \mbox{and therefore} \;\; \sum_{i=1}^n m_i
< \sum_{j=1}^k f_j \label{eq:lowden4}
\end{equation}

\begin{thm}\label{thm:totfslowden}
Algorithm \ref{alg:leftrightshift} connects the total free space
with at most $2n$ moves and uses $O(n)$ computing time.
\end{thm}

\begin{proof}
The number of shifts and the computing time are obvious. We will
show that at the end of the first loop, the rightmost free interval is
greater than any module and therefore {\em all} modules can be
shifted to the right in the second loop.

Let $F_1,\ldots,F_k$ denote the free intervals in $L$ at the end of the
first loop. Then every $F_i$, $i\in \{1,\ldots,k-1\}$ is bounded to
the right by a module $M_{j}$ with $m_{j} > f_i$ (otherwise $m_j$
could be shifted). If this holds for $F_k$ as well, we can
conclude that $\sum_{i=1}^k f_i < \sum_{i=1}^n m_i$, which
contradicts (\ref{eq:lowden4}). Hence, there is no module to the right
of $F_k$, and we get with $m^\star = \max_{1,\ldots,n} \{m_i\}$

%\vspace*{-5mm}
\begin{equation*}
\frac{l}{2} + \frac{1}{2} m^\star %\max_{1,\ldots,n} \{m_i\}
\stackrel{(\ref{eq:lowden3})}{\leq} \sum_{i=1}^k f_i
%\stackrel{(\ref{eq:lowden4})}{<}
< \sum_{i=1}^n m_i + f_k
\stackrel{(\ref{eq:lowden2})}{\leq}
\frac{l}{2} -
\frac{1}{2}m^\star + f_k  %\max_{1,\ldots,n} \{m_i\} + f_k,
\end{equation*}
implying
%\begin{equation*}
%\max_{1,\ldots,n} \{m_i\} < f_k \;.
$m^\star  < f_k$.
% \end{equation*}
\end{proof}

\section{A Quadratic Lower Bound}
As a consequence of the hardness and inapproximability results we
focus on developing
heuristic approaches for the MDP. 
In this section, we bound
the number of steps needed by any algorithm that constructs a
maximum free interval, even in the homogeneous version. 
In the next sections, we state a heuristic and give
experimental results.

\begin{thm}
There is an instance of the maximum defragmentation problem such
that any algorithm needs at least $\Omega(n^2)$ steps to solve it.
\end{thm}

\begin{proof}
We construct the instance in the following way. For an even number $n$,
we place $n$ modules,
indexed from left to right by $1,\ldots,n$. 
The sizes of the modules are $m_j = m_{n+1-j} = n+2-2j$
for $1\leq j \leq \frac{n}{2}$. $M_1$ has a free interval of size
$1$ to its left and $M_n$ has a free interval of size
$1$ to its right.
In addition, every pair of consecutive modules is separated by a
free interval of size one, except for the pair $M_{\frac{n}{2}}$ and
$M_{\frac{n}{2}+1}$, which is separated by a distance of two. In
this initial configuration we denote the free intervals by
$F_1,\ldots,F_{n+1}$, and their sizes by $f_1,\ldots,f_{n+1}$. 
Fig.~\ref{fig:lbnsquare} shows an example for $n=8$.

\begin{figure}[ht]
\centering
\input{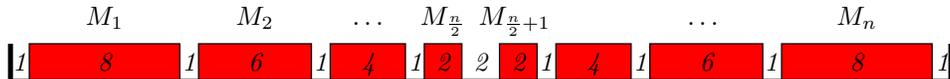}
\caption{\label{fig:lbnsquare} The instance for $n=8$.}
\end{figure}

The following properties of this instance are essential for the rest
of the proof:

\begin{list}{}{}

\item[(i)] The module sizes $m_j = m_{n+1-j} = n+2-2j$, $1\leq j \leq
\frac{n}{2}$, are even.

\item[(ii)] $m_j = m_{n+1-j} = \sum_{i=j+1}^{n+1-j} f_i$ holds for
any pair $M_j$, $M_{n+1-j}$ (i.e., the total free interval between two
modules of equal size is equal to the modules' sizes).

\item[(iii)] Every module has to be moved at least once (because of
the small free intervals at the left and right end of $L$).

\item[(iv)] For  $i < j \leq \frac{n}2$, we have
$m_i = m_{n+1-i} \geq m_{n+1-j} =m_j$; in particular, this
means that the pair $M_i$, $M_{n+1-i}$ can be
moved only if $M_j$, $M_{n+1-j}$ can be moved.

\end{list}

In the beginning, only the modules $M_{\frac{n}{2}}$ and
$M_{\frac{n}{2}+1}$ with $m_{\frac{n}{2}}=m_{\frac{n}{2}+1}=2$ can
be moved. Using three moves, a free interval of size four can be
constructed. Note that both modules have to be moved and that there
can be a free interval of size four only if there is exactly one free interval 
between $M_{\frac{n}{2}-1}$ and $M_{\frac{n}{2}+2}$. Now we show
by induction for $j$ from $\frac{n}{2}$ to $1$ that
\begin{list}{}{}

\item[(a)] at least $n-2j+2$ steps are necessary to make a pair
$M_{j-1}$, $M_{n+1-(j-1)}$ movable after the pair $M_{j}$,
$M_{n+1-j}$ became movable and

\item[(b)] in the situation in which $M_{j-1}$, $M_{n+1-(j-1)}$
become movable, there is exactly one free interval between these two
modules.

\end{list}

Both properties clearly hold for $j=\frac{n}{2}$ and we assume that
$M_j$ and $M_{n+1-j}$ for $1 \leq j < \frac{n}{2}$ became movable (for
the first time) by the last step.

By part (b) of the induction hypothesis, the modules and free intervals
in the area between $M_{j-1}$ and $M_{n+1-j}$ are currently arranged
in the following order, described from left to right: a free interval 
of size one, a sequence of modules, a free interval of size $m_j$, a
sequence of modules, and a free interval of size one (see 
Fig.~\ref{fig:lbnsquare2}). The modules in the rest of $L$ are still in
their initial position (otherwise $M_j$ and $M_{n+1-j}$ could have
been moved earlier because of (iv)).

\begin{figure}[ht]
\centering
\input{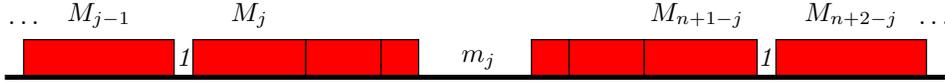}
\caption{\label{fig:lbnsquare2} The situation when $M_j$ and
$M_{n+1-j}$ can be moved for the first time.}
\end{figure}

Property (b) is a straightforward implication of (ii) and we show
that (a) holds as well. Suppose for a contradiction that $M_{j-1}$
and $M_{n+2-j}$ can be made movable without shifting or ``jumping'' a module
$M_k$ with $j \leq k \leq n+1-j$, i.e., without moving a modules that lies
between $M_{j-1}$ and $M_{n+2-j}$). We assume w.l.o.g.\ that $M_k$ is in
the same sequence as $M_j$. Thus, the distance from $M_k$'s left boundary
to the right boundary of $M_{j-1}$ can be calculated as the sum of
the sizes of modules lying on the left side of $M_k$ plus one. By
(i), this is an odd number. The same holds for the distance from
$M_k$'s right boundary to the left boundary of $M_{n+2-j}$. Again
using (i), this implies that none of these intervals can completely be
filled with other modules. Hence, by (ii), $M_{j-1}$ and $M_{n+2-j}$ 
can never be moved without moving $M_k$. There are $n-j+ 2 -j-1
+1 = n-2j+2$ modules initially placed between $M_{j-1}$ and
$M_{n+2-j}$ and each of them has to be moved.

Altogether, this implies a lower bound of
$\sum_{j=1}^\frac{n}{2}(n-2j+2) = \frac{n^2}{4} + \frac{3n}{2}$ on
the total number of steps.
\end{proof}

\section{A Heuristic Method}
\label{sec:heuristic}

For runtime defragmentation, we propose a tabu search with a tabu list of length
$\frac{n}{2}$, see Algorithm~2. In every iteration, all homogeneous modules $M_i$ are
moved to the left end and to the right end of the free intervals
that are greater than or equal to $m_i$. All inhomogeneous modules
are moved to any feasible position. Each move is evaluated by a
fitness function that divides the size of the maximal free interval by
the number of free slots. The move yielding the
configuration with the highest fitness is chosen. Ties are broken by
choosing the first one. The resulting configuration is added to the
tabu list.

%TODO: tabu pseudocode

If the current solution is the best one found so far, it is stored.
The heuristic ends if either a fitness of $1.0$ (i.e., optimality)
is achieved or $2n^2$ iterations have been performed. As seen above,
there are instances for which $\Omega(n^2)$ moves are necessary. 
Moreover, we conjecture that the number of necessary moves is 
in $\Theta(n^2)$. 
\begin{algorithm}[p]\caption{Tabu Search}
%% \begin{figure}[t]
%% \framebox[\textwidth][l]{{\bf Algorithm 2:} Tabu Search}
\begin{tabbing}
\mbox\qquad\=\qquad\=\qquad\=\qquad\=\kill
{\it counter} := 0;\\
{\it tabulist} := \{\};\\
{\it maxfitness} := 0.0;\\

{\bf while} ({\it counter} $\leq 2n^2$) and ({\it maxfitness} $< 1.0$) {\bf do}\\
\> {\bf foreach} module $M_i$ in the array {\bf do}\\
\>\> {\it storedfitness} := 0.0\\
\>\> {\bf if} $M_i$ is homogeneous {\bf then}\\
\>\>\> {\bf foreach} free interval $F_j$ {\bf do}\\
\>\>\>\> move $M_i$ to the left end of $F_j$\\
\>\>\>\> evaluate move\\
\>\>\>\> move $M_i$ to the right end of $F_j$\\
\>\>\>\> evaluate move\\
\>\>\>\> move $M_i$ back to its original position\\
\>\>\> {\bf end foreach}\\
\>\> {\bf else}\\
\>\>\> {\bf foreach} position $P$ for $M_i$ that is feasible and not blocked by another module {\bf do}\\
\>\>\>\> move $M_i$ to $P$\\
\>\>\>\> evaluate move\\
\>\>\>\> move $M_i$ back to its original position\\
\>\>\> {\bf end foreach}\\
\>\> {\bf end if}\\
\>\> {\bf if} {\it storedfitness} $>$ {\it maxfitness} {\bf then}\\
\>\>\> apply {\it storedmove}\\
\>\>\> store {\it storedmove} in {\it tabulist}\\
\>\>\> {\it maxfitness} := {\it storedfitness}\\
\>\>{\bf end if}\\
\> {\bf end foreach}\\
\> {\it counter}++\\
{\bf end while}\\
\\
{\bf Procedure} evaluate move:\\
\> {\bf if} move is not stored in {\it tabulist}\\
\>\> {\it thisfitness} $:=$ size of the maximal free interval / number of free slots \\
\>\> {\bf if} {\it thisfitness} $>$ {\it storedfitness} {\bf then}\\
\>\>\> store move in {\it storedmove}\\
\>\>\> {\it storedfitness} := {\it thisfitness}\\
\>\> {\bf end if}\\
\> {\bf end if}
%% \mbox{}\hrulefill
%% \end{figure}
\end{tabbing}
\end{algorithm}

\section{Experimental Results}
\label{sec:experiments}

\subsection{Compacting an FPGA}
We performed a series of experiments for defragmentation based on scenarios of FPGAs with and 
without heterogeneities and different densities (i.e., different ratios of
occupied space compared to unoccupied space). Fig.~\ref{max:cf1} shows the results
for two FPGAs, both having 94
slots. The first FPGA does not contain any heterogeneities, while
the second one is an FPGA with heterogeneities at 
positions 3, 24, 45, 50, 71 and 82. Moreover, we compared our heuristic
to a simple greedy approach that moves every module to the most
promising position (i.e., to the position for which the ratio of
the size of the maximal free interval and 
the size of the total free space is maximal).

Generating the input was done in two steps, depending on the
size of the maximal free interval $F^\star$. In the first step
the module size is chosen with equal probability from the set
\{$1,\ldots, f^\star$\}. This ensures that the modules can be
inserted. The exact position is chosen again with equal probability
among all feasible positions. If the interval
occupied by the module contains an heterogeneity, this heterogeneity
is assigned to the corresponding position of the module. The size of
the first module is shrunk by a factor of $0.6$ in order to ensure
that it can be moved.

For the density ranging from $0.3$ to $0.9$ with steps of size
$0.05$, we performed 100 runs of the tabu search and the greedy strategy 
for each value and
took the average value of the {\em number of free intervals} and
the {\em size of the maximal free interval}. The results are shown
in %Figures \ref{max:cf0} and 
Fig.~\ref{max:cf1}. The diagrams
show the size of the maximal free interval (top row) of the
array and the number of free spaces (bottom row) 
before and after the defragmentation. In the
array with no heterogeneities (left column),
there is an improvement of up to 40\%.
On the FPGA (right column)
the size of any maximal free interval is limited
to $20$ slots due to the heterogeneities. For a density of less than
$\frac{1}{2}$, the tabu search achieves this upper bound for almost
all instances. For larger densities, it achieves an improvement of
approximately 35\%.

The change in the number of free intervals before and after
defragmentation is displayed in the right charts of Fig.~\ref{max:cf1}.
In the array with no heterogeneities there is an
increase of 50\%. For the FPGA there is almost no improvement for low densities (less than $\frac{1}{2}$) and an improvement of approximately 25\% for larger ones.

\begin{figure*}[p]
 \mbox{}\hfill 
  {\epsfig{figure=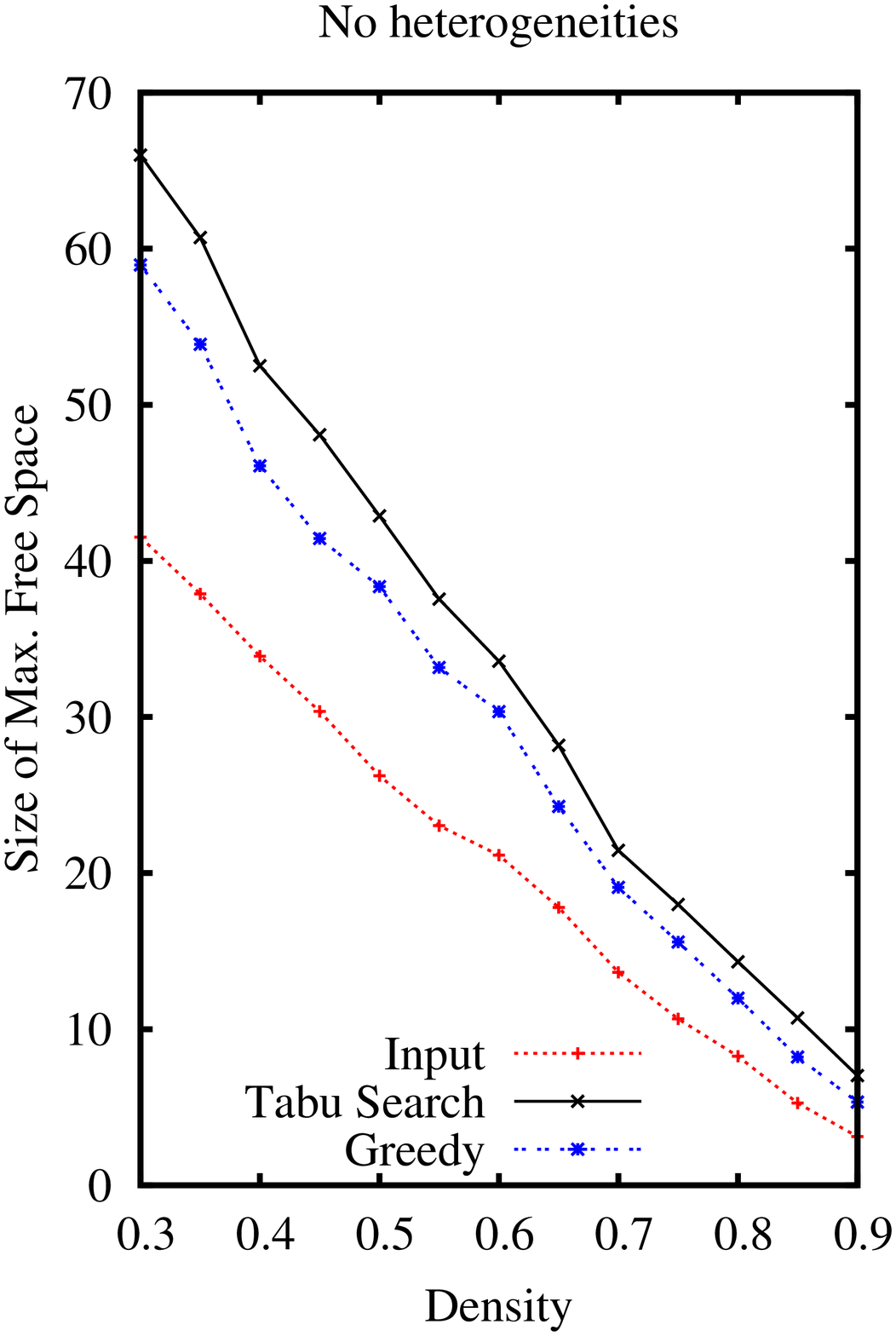,width=67mm}}
\hfill
  {\epsfig{figure=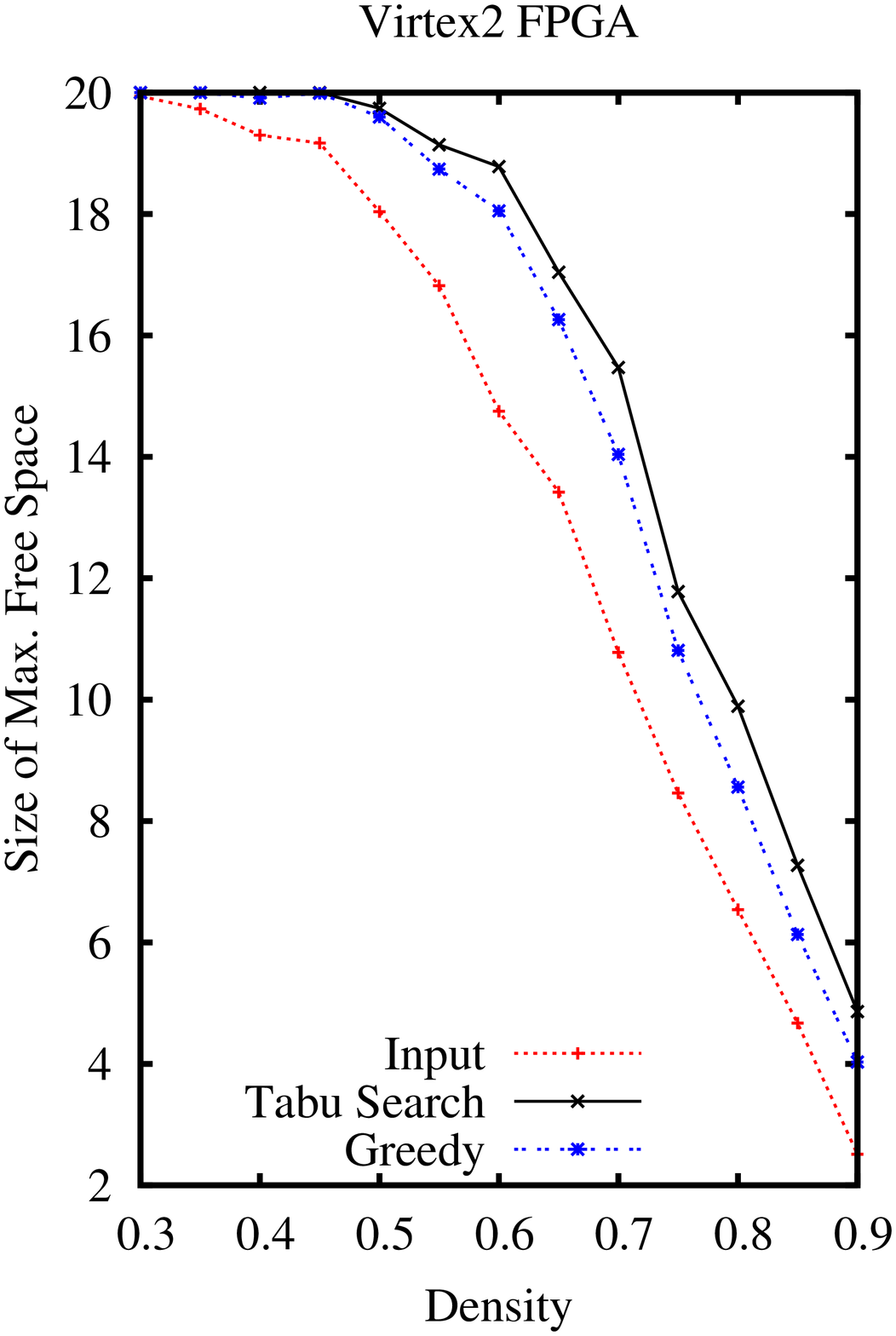,width=67mm}} 
\hfill\mbox{}\\
 \mbox{}\hfill  
{\epsfig{figure=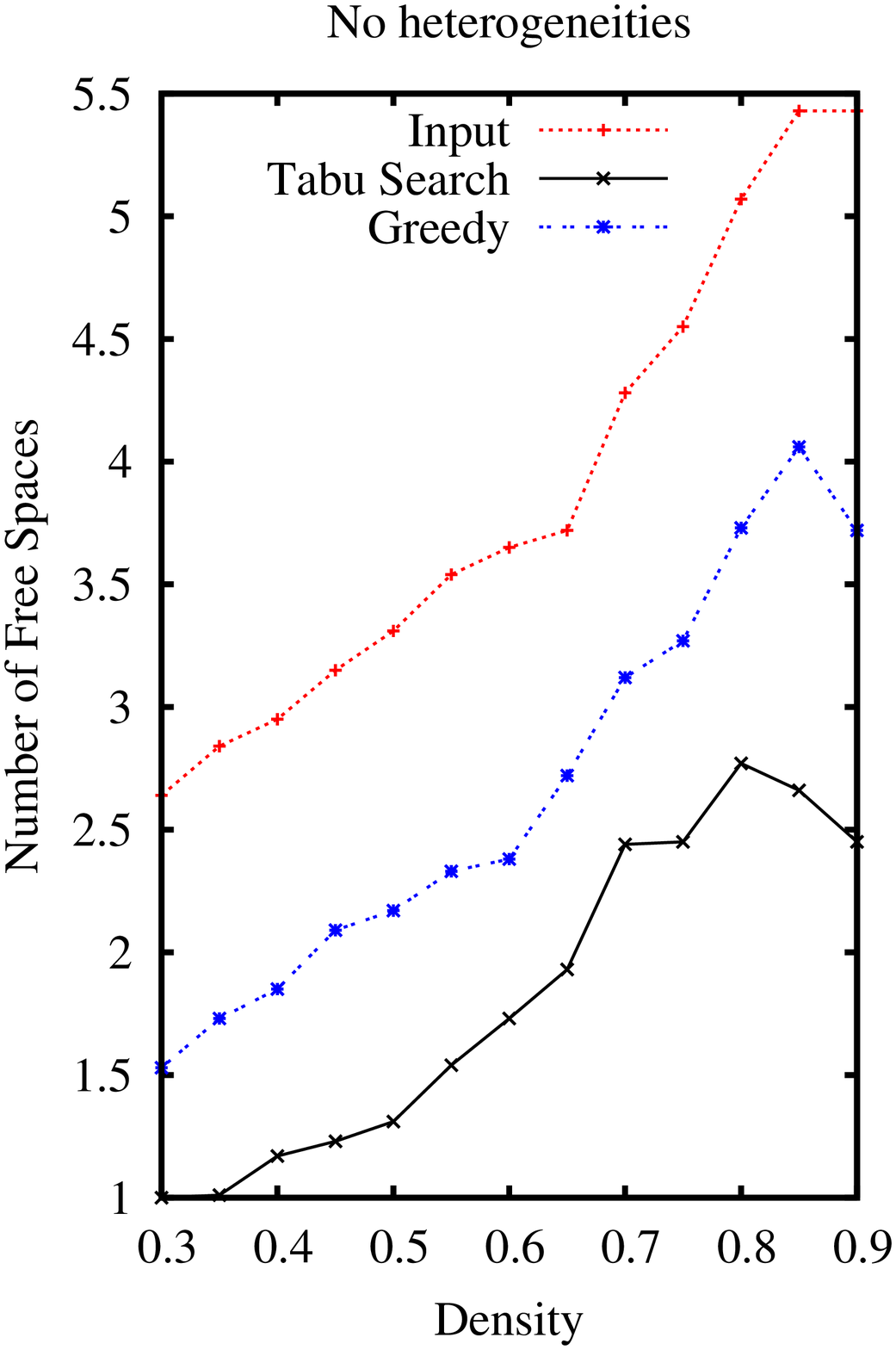,width=67mm}}
\hfill
  {\epsfig{figure=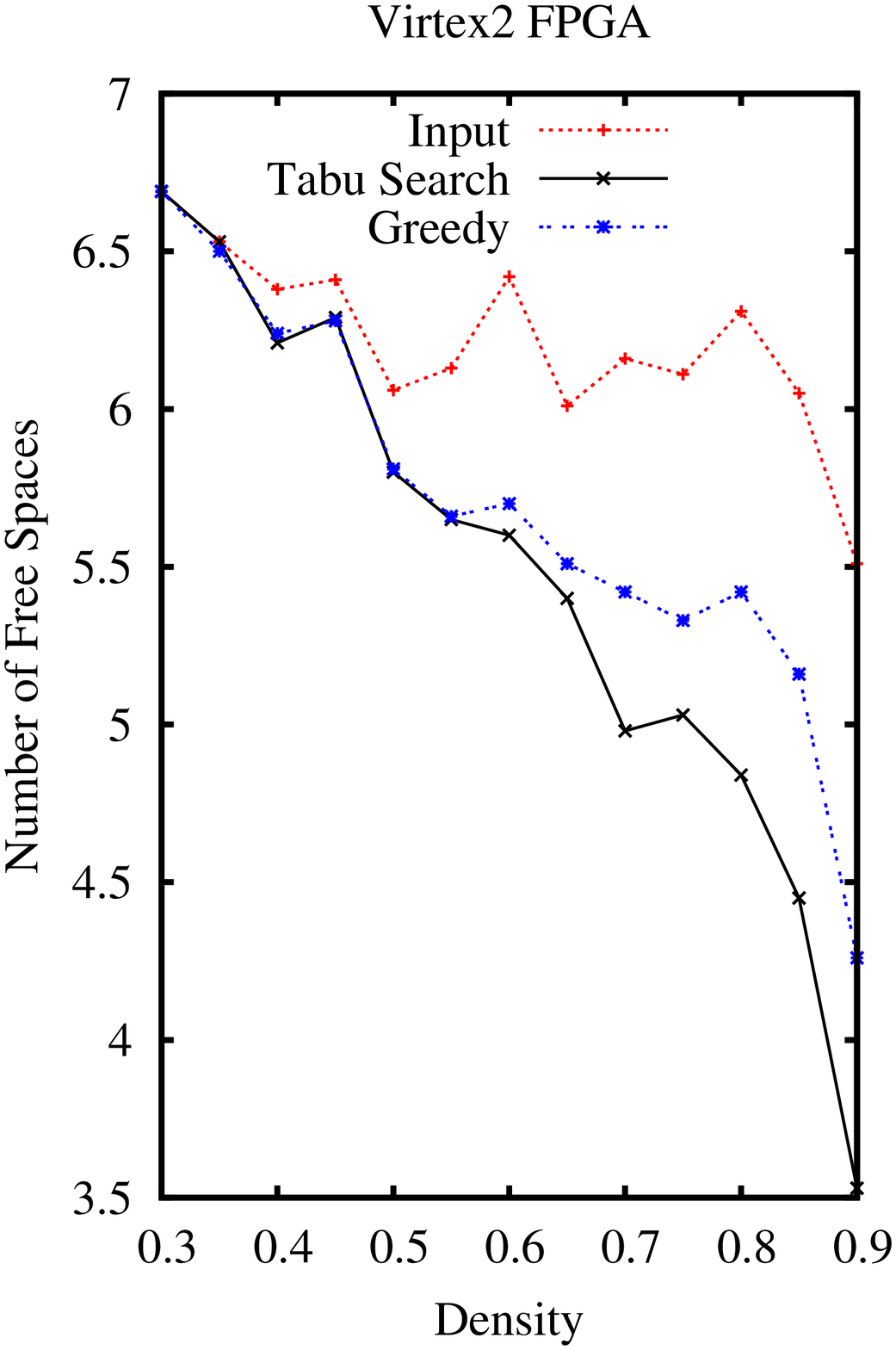,width=67mm}}
\hfill\mbox{}
\caption{Averages over 100 runs. (Top left)
Size of the maximal free interval before and after
defragmentation, using our heuristic and a simple greedy approach
in an array with no heterogeneities.\newline
(Top right)
Size of the maximal free interval before and after
defragmentation of the FPGA. \newline
(Bottom left)
Number of free intervals before and after defragmentation in an array with no heterogeneities.\newline
(Bottom right)
Number of free intervals before and after defragmentation of
the FPGA.\label{max:cf1}}
\end{figure*}

% \begin{figure}[!h]
% \centerline{\epsfig{figure=PLOT-MAX-CF0a.eps,width=50mm}}
% \vspace*{-3mm}
% \caption{Size of the maximal free space before and after
% defragmentation in an array with no heterogeneities.}\label{max:cf0}
% \vspace*{-7mm}
% \end{figure}
%\begin{figure}[!h]
  %\centerline{\epsfig{figure=PLOT-MAX-CF1a.eps,width=50mm}}
%\caption{Size of the maximal free space before and after
%defragmentation of the Virtex 2 FPGA.}\label{max:cf1}
%\end{figure}
%
%\begin{figure}[!h]
  %\centerline{\epsfig{figure=PLOT-NUM-CF0a.eps,width=50mm}}
%\caption{Number of free intervals before and after defragmentation in an
%array with no heterogeneities.}\label{num:cf0}
%\end{figure}
%
%\begin{figure}[!h]
  %\centerline{\epsfig{figure=PLOT-NUM-CF1a.eps,width=50mm}}
%\caption{Number of free intervals before and after defragmentation of
%the Virtex 2 FPGA.}\label{num:cf1}
%\end{figure}

\subsection{Case Study}
\label{sec:casestudy}

In this section, a case study is given that demonstrates the
efficiency of the proposed techniques and how they can be applied to a
real-world scenario. We assume a dynamically partially
reconfigurable device, whose reconfigurable area is separated into
$94$ columns, also called {\em slots}. Modeling typical FPGAs, some of these slots contain
no logic resources, but a heterogeneities such as
BlockRAMs. This setting is illustrated in Figure~\ref{fig:casestudy_initial1}.

Furthermore, assume that one or multiple applications with a collection of
modules are executed on this device; e.g., these could a video
processing and a number cruncher application whose current state can
rather easily be saved and restored at a different position on the
reconfigurable device with moderate costs. During the execution of the
applications, different modules finish and are removed, while new modules
need to be placed. Thus, the free space on the reconfigurable device
can be scattered over the whole reconfigurable area. This situation is
illustrated in the upper part of Figure~\ref{fig:casestudy_greedy}.

\begin{figure}[p]
\centering
%\vspace{-5 mm}
\includegraphics[width=.8\linewidth]{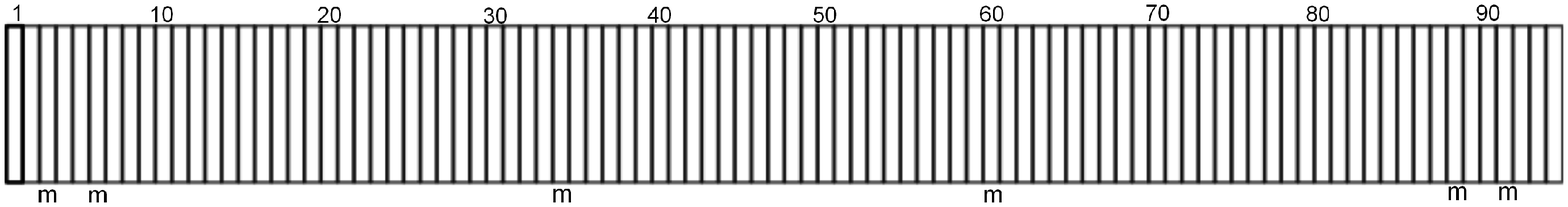}
\caption{\label{fig:casestudy_initial1} Initial state in example scenario. The
reconfigurable device consists of $94$ partially reconfigurable, empty slots.
Those containing a heterogeneity, such as BlockRAMs, are marked below with an
``m''.} \end{figure}

\begin{figure}[p]
\centering
%\vspace{-5 mm}
\includegraphics[width=.8\linewidth]{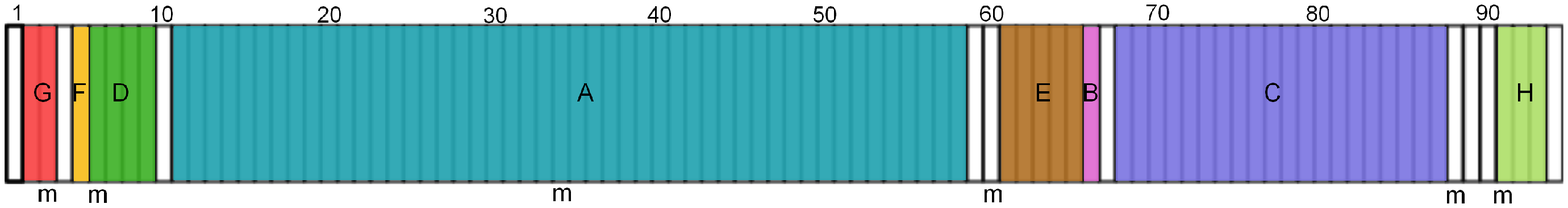}
\includegraphics[width=.8\linewidth]{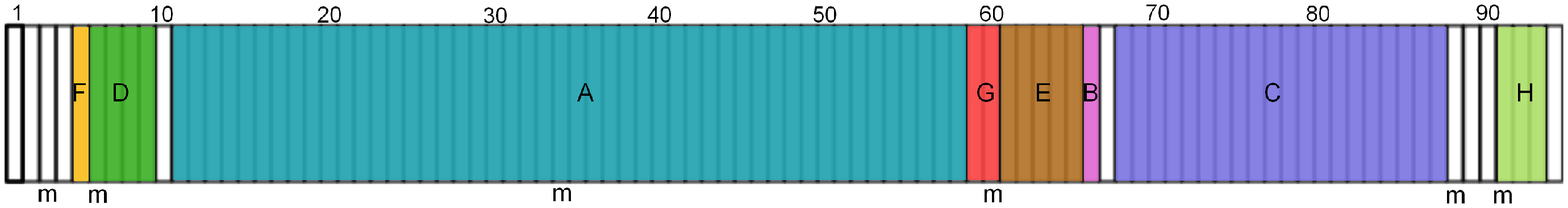}
\caption{\label{fig:casestudy_greedy} (Top) Fragmented state in example scenario. Free space is scattered over the whole reconfigurable area. 
(Bottom) Free space after defragmentation with greedy approach.}
\end{figure}

\begin{figure}[p]
\centering
%\vspace{-5 mm}
\includegraphics[width=.8\linewidth]{DefragCaseStudy_initial2.eps}
\includegraphics[width=.8\linewidth]{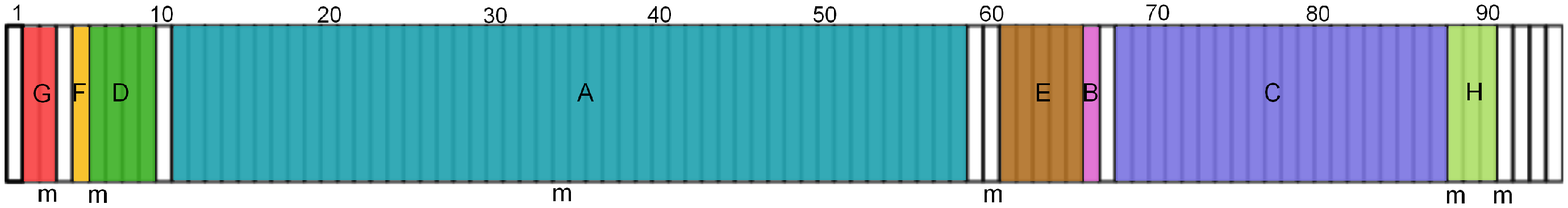}
\includegraphics[width=.8\linewidth]{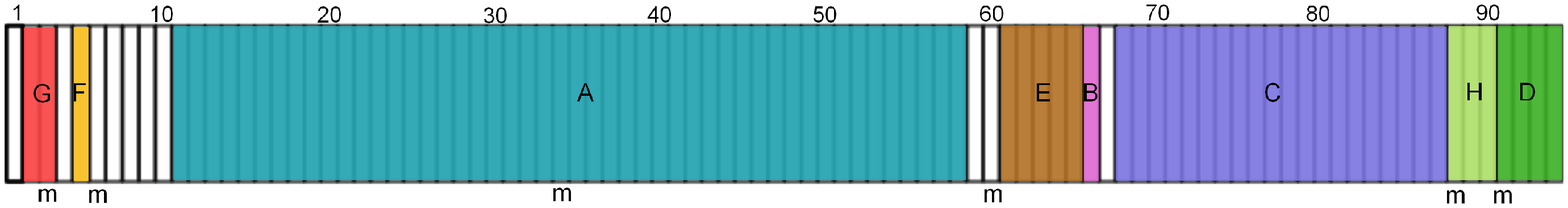}
\includegraphics[width=.8\linewidth]{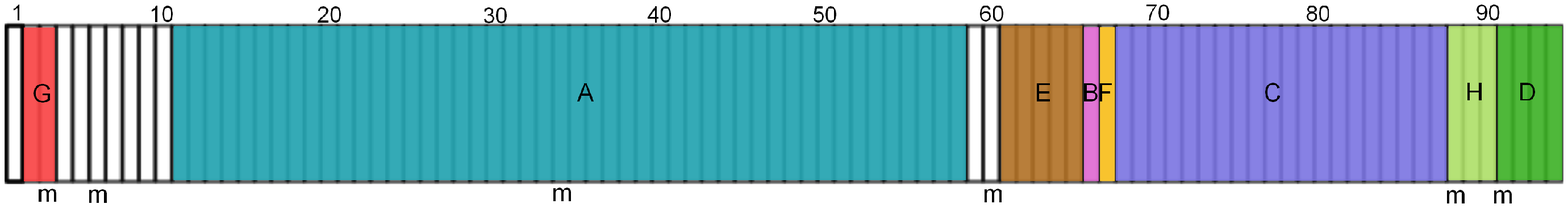}
\includegraphics[width=.8\linewidth]{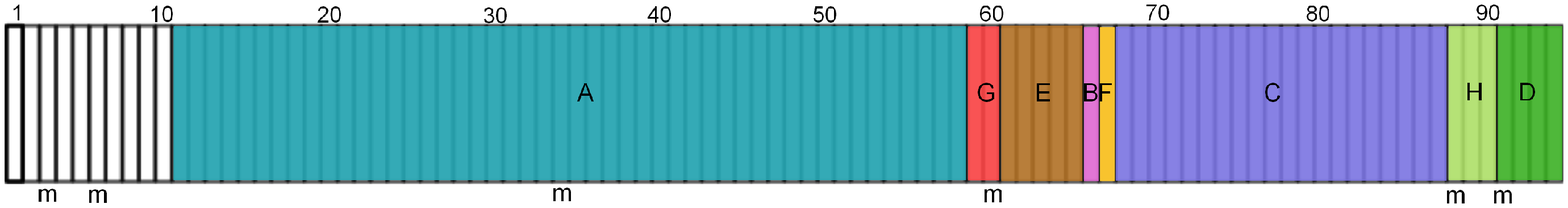}
\caption{\label{fig:casestudy_ts} Tabu-search algorithm: defragmentation in four move steps (shown from top to bottom).}
\end{figure}

The fragmented free space on the reconfigurable area is a common,
unavoidable scenario, for which our proposed defragmentation
techniques represent an applicable and efficient solution. Our first
approach, the greedy algorithm, selects in each setting a step that
optimizes the resulting maximal contiguous free space. Based on the
state of the example in the upper part of Figure~\ref{fig:casestudy_greedy}, the
greedy algorithm moves the module ''G'' to position $59$. Thus, a
biggest free space is achieved within a single move. The
heterogeneity requirements of module ''G'' are fulfilled at the time
at this position: at position $60$, BlockRAMs are provided for the
right part of the module. Afterwards, no single move that
provides an improvement on the maximum free contiguous space is possible. Thus, the
greedy algorithm terminates. Note that the evaluation of each
possible step in the algorithm checks the maximal free space by taking
into account all contiguous free slots, no matter if they contain
heterogeneities or not.

In our second approach, the maximum free contiguous space is optimized
using tabu search, see Figure~\ref{fig:casestudy_ts}. Based on the state of the same example, module ''H'' is relocated to slot
position $88$. This step yields a maximal, contiguous free space of four slots
including a single BlockRAM heterogeneity slot. In a
second step, module ''D'' is moved to slot $91$, where module ''H''
was located before. Thus, a new maximal free space is created starting
at slot $6$ up to $10$. All other steps would have created a free
contiguous space with a size less than $5$ slots. Further, this is also the
only position to which module ''D'' can be moved, due to its heterogeneity
constraints. In a third step, module ''F'' is moved to the single
empty slot without BlockRAMs between module ''B'' and module ''C''. In
a next step, module ''G'' can be either moved to slot position $6$ or
to slot position $59$; both satisfy its heterogeneity
demands. Finally, it is moved to the latter position, because this
results in a maximal free space of $10$ slots. It is also possible
that multiple single steps offer the same increase in contiguous free
space; 
in our current implementation,  one single move is selected randomly.

When the greedy algorithm is applied to the example input, a
contiguous free space of four slots is achieved. In contrast, the
tabu search merges all free space and yields one single contiguous 
block of free space of size $10$. This shows
the usefulness of defragmentation techniques, and the importance of
the corresponding strategy. Similar scenarios of scattered empty space
and heterogeneities on the reconfigurable device are common when
executing modules. New modules with big area requirements must
unnecessarily be delayed without defragmentation steps, which can be
avoided with appropriate defragmentation strategies. How far different
strategies can deviate is shown by comparing the results of the greedy
and the tabu-search approach for this example.

\subsection{Makespan}

\newlength{\makespanfig}
\setlength{\makespanfig}{60mm}
\def\makespanspace{\hspace{-5mm}}

\begin{figure*}[p]
\mbox{}\makespanspace
{\epsfig{figure=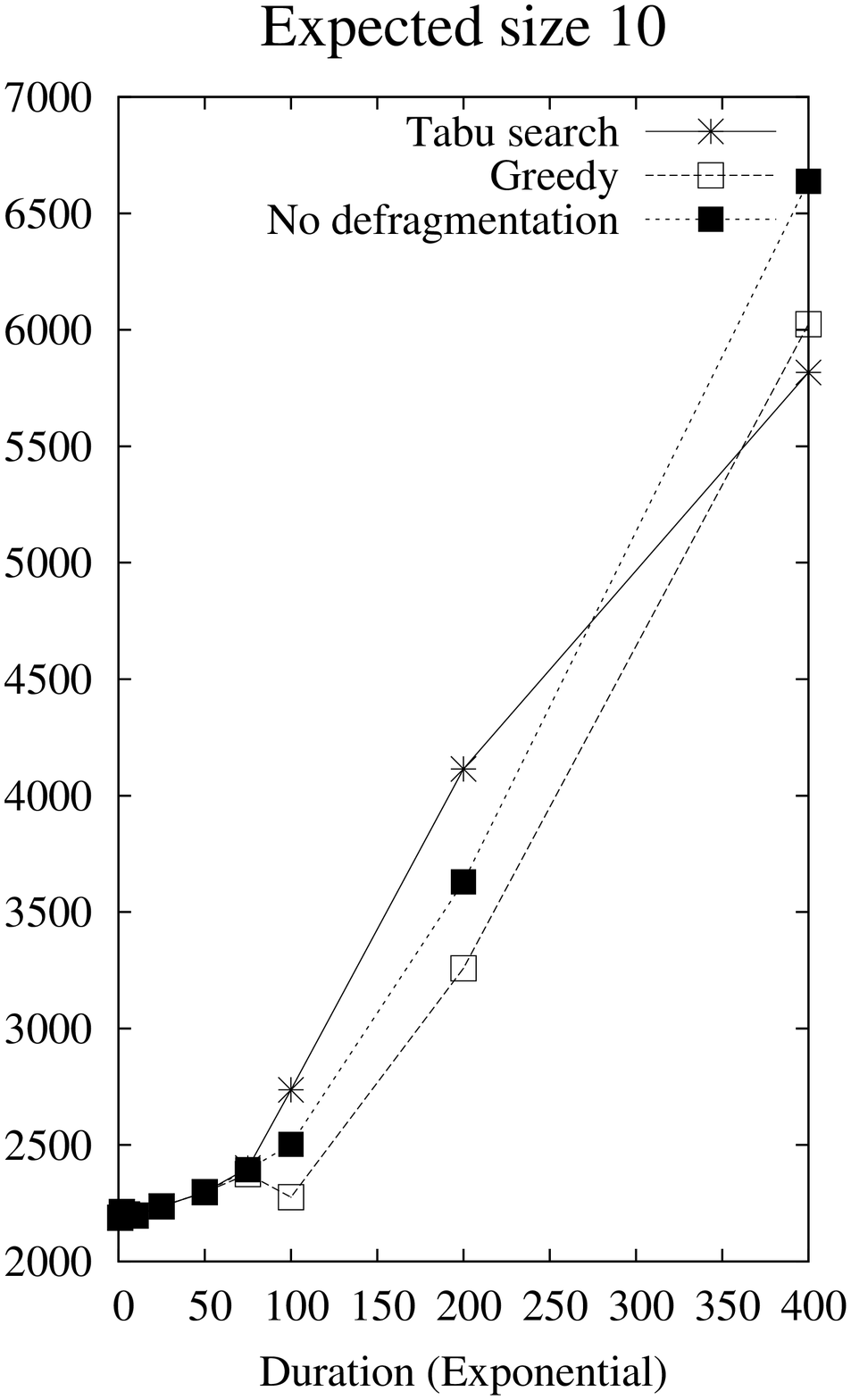,width=\makespanfig}}\makespanspace
{\epsfig{figure=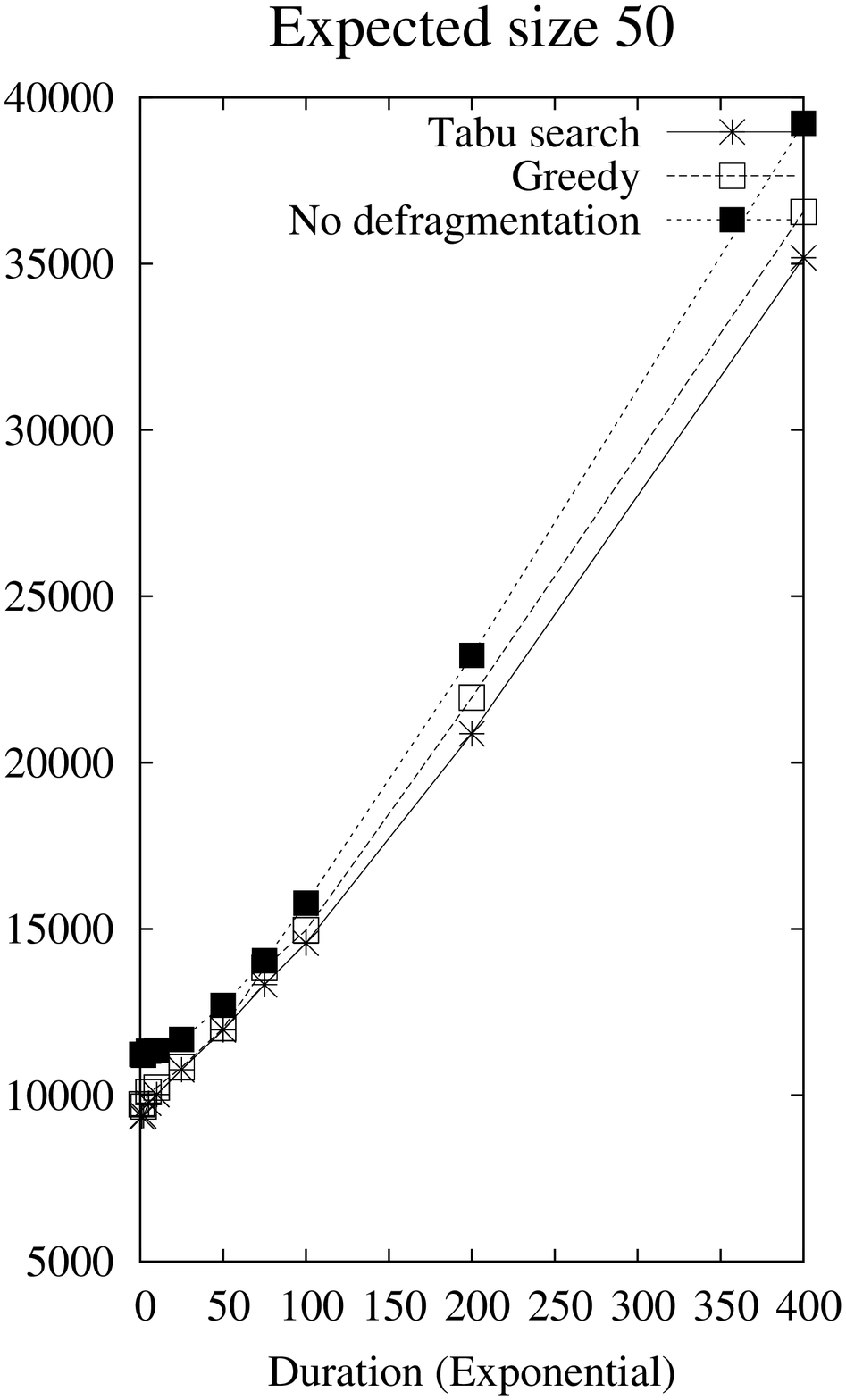,width=\makespanfig}}\makespanspace
%% \mbox{}\hfill {\epsfig{figure=PLOT-4-Makespan-CF0-Sn0100Texxxxxxx.eps,width=\makespanfig}}\makespanspace
{\epsfig{figure=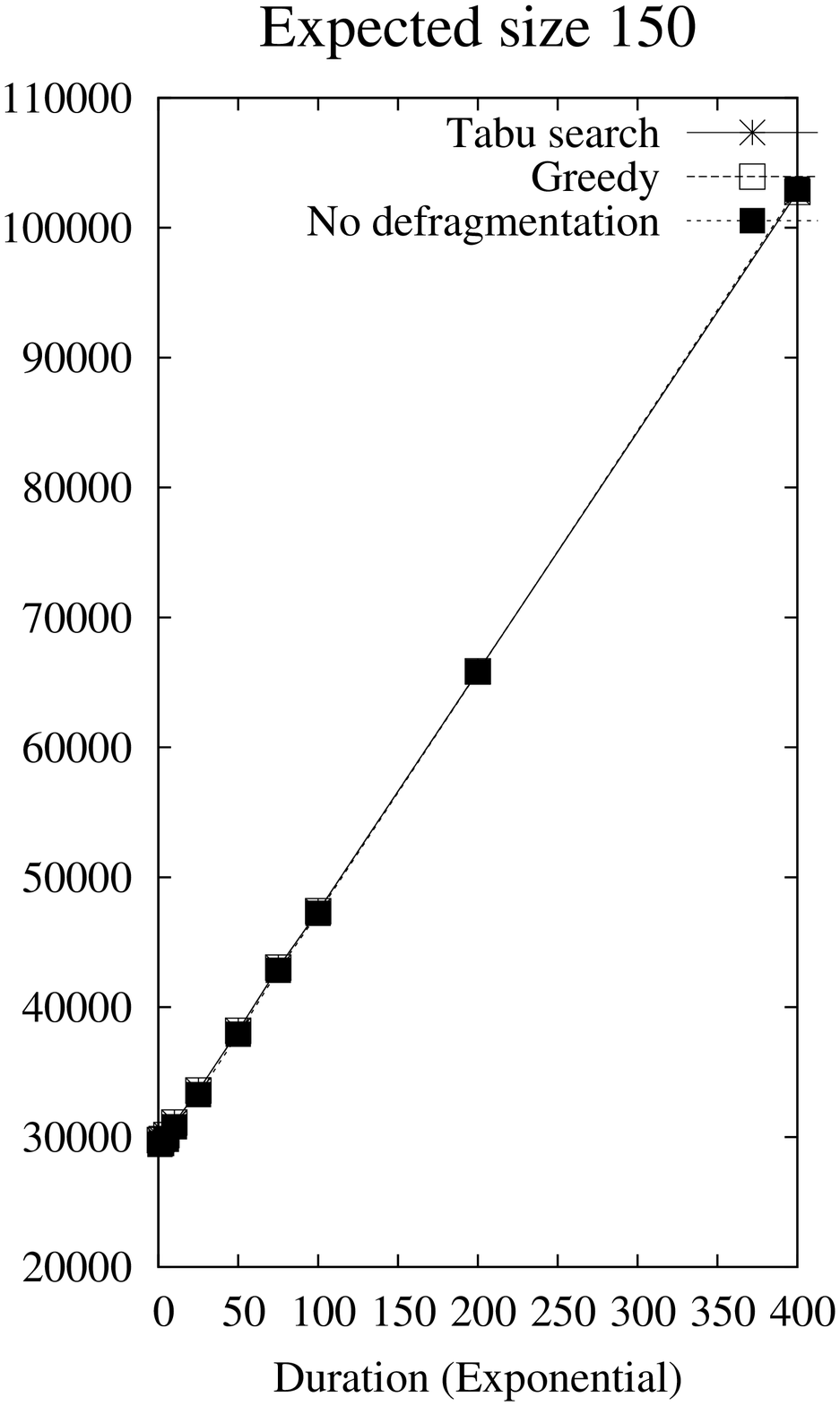,width=\makespanfig}}\makespanspace
%% {\epsfig{figure=PLOT-3-Makespan-CF0-Sh0010Thxxxxxxx.eps,width=\makespanfig}}\makespanspace
%% {\epsfig{figure=PLOT-3-Makespan-CF0-Sh0050Thxxxxxxx.eps,width=\makespanfig}}\makespanspace
%% % {\epsfig{figure=PLOT-3-Makespan-CF0-Sh0100Thxxxxxxx.eps,width=\makespanfig}}\makespanspace
%%  {\epsfig{figure=PLOT-3-Makespan-CF0-Sh0150Thxxxxxxx.eps,width=\makespanfig}}\makespanspace
\mbox{}

\caption{Comparison of makespans for schedules using tabu search, greedy, and no fragmentation for an array of size $\ell=200$.
The average module size is fixed to 10, 50, and 150 columns,
the average duration time ranges from 1 to 400 time units.
The $y$-axis shows the total makespan in time units.
\label{smalltimes-fig}}
\end{figure*}

\begin{figure*}[]
\mbox{}\makespanspace
{\epsfig{figure=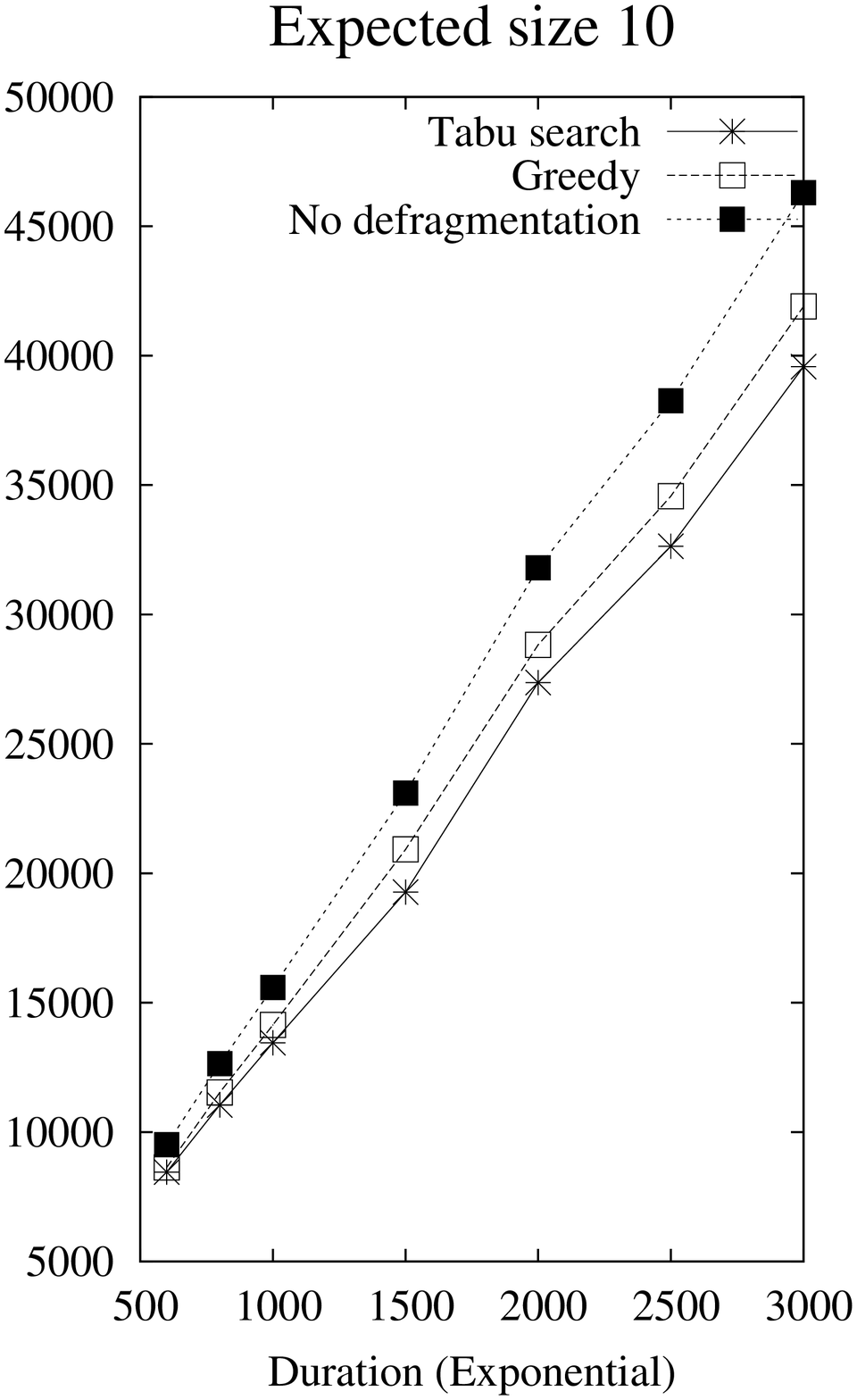,width=\makespanfig}}\makespanspace
{\epsfig{figure=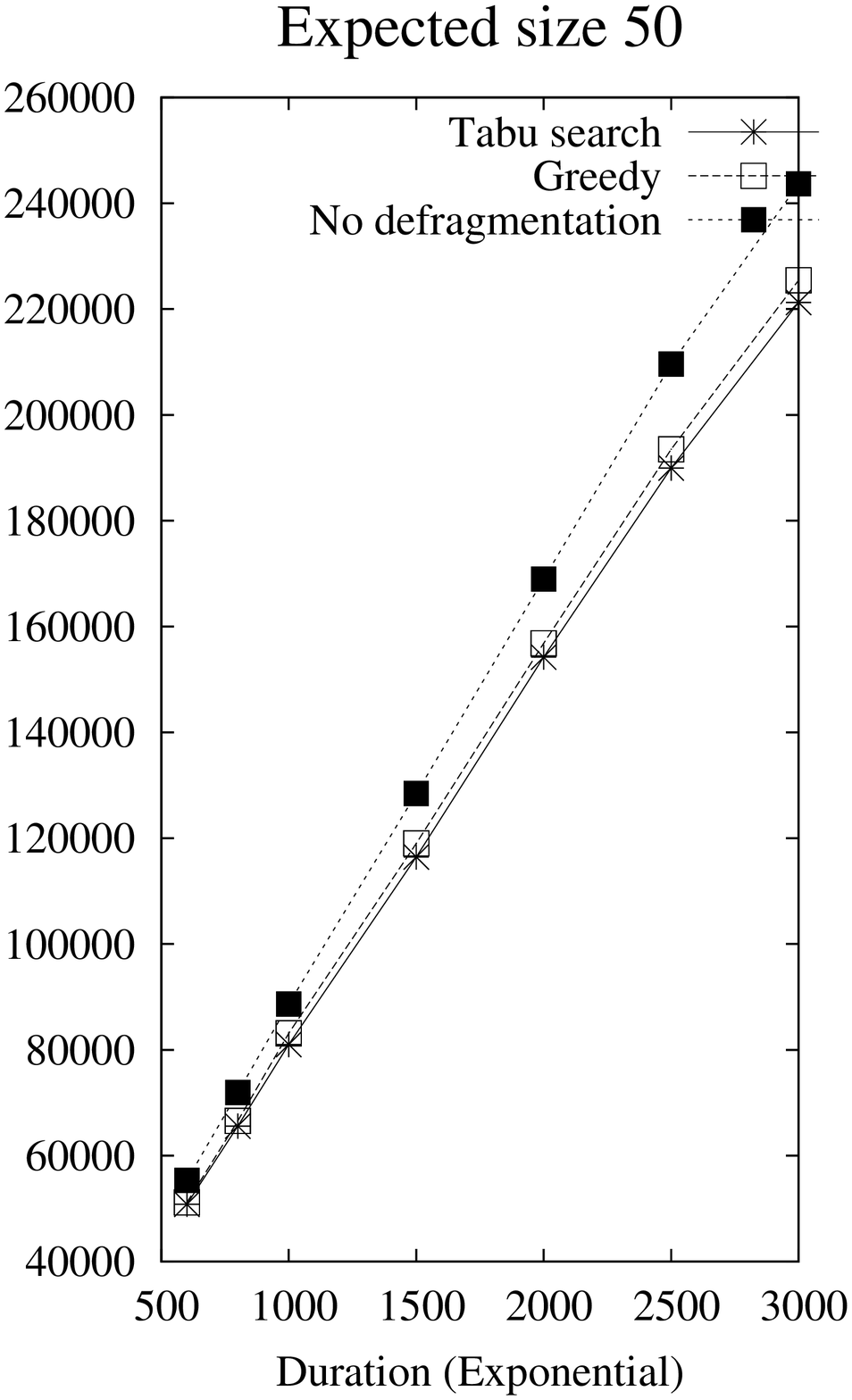,width=\makespanfig}}\makespanspace
{\epsfig{figure=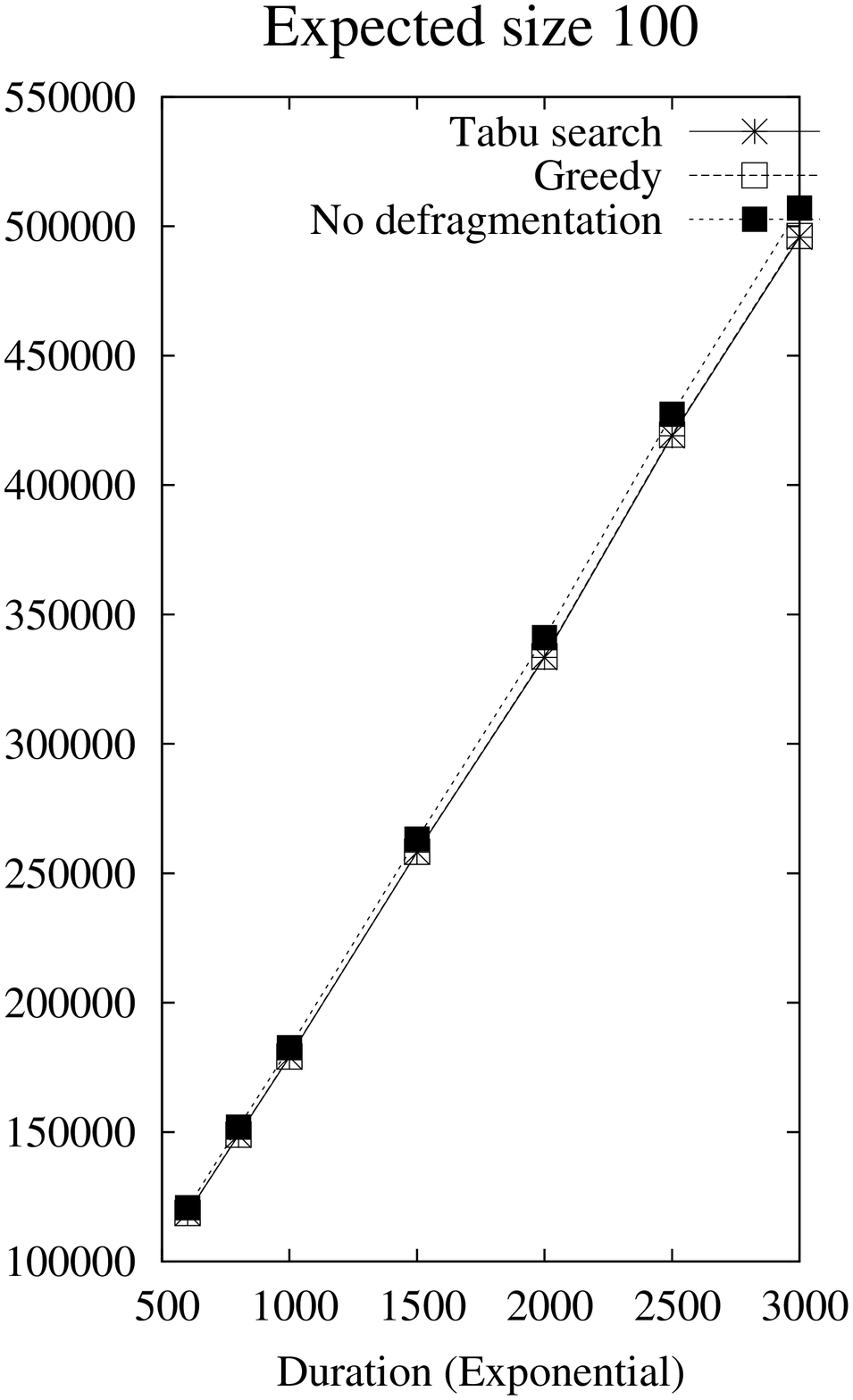,width=\makespanfig}}\makespanspace
%{\epsfig{figure=PLOT-2-Makespan-CF0-Sn0150Texxxxxxx.eps,width=\makespanfig}}\makespanspace
%%  {\epsfig{figure=PLOT-5-Makespan-CF0-Sh0010Thxxxxxxx.eps,width=\makespanfig}}\makespanspace
%%  {\epsfig{figure=PLOT-5-Makespan-CF0-Sh0050Thxxxxxxx.eps,width=\makespanfig}}\makespanspace
%%  % {\epsfig{figure=PLOT-5-Makespan-CF0-Sh0100Thxxxxxxx.eps,width=\makespanfig}}\makespanspace
%%  {\epsfig{figure=PLOT-5-Makespan-CF0-Sh0150Thxxxxxxx.eps,width=\makespanfig}}\makespanspace
\mbox{}
\caption{Comparison of makespans for schedules using tabu search, greedy, and no fragmentation for an array of size $\ell=200$.
The average module size is fixed to 10, 50, and 100 columns,
the average duration time ranges from 600 to 3000 time units.
The $y$-axis shows the total makespan in time units.
\label{bigtimes-fig}}
\end{figure*}

We also simulated the impact on the total makespan (i.e., the total
execution time) by randomly generating sequences of modules. A
sequence consists of 200 modules, for each module we chose size and
duration randomly using different distributions. 
%It turned out that
% the type of random distribution has little impact on the results.
Fig.~\ref{smalltimes-fig} and Fig.~\ref{bigtimes-fig} show examples
in which the size was chosen by normal distribution and duration according to an 
exponential distribution.  We used the exponential distribution for
the duration time, because this distribution models typical life
times~\cite{bn-gs-95}.
% both size and duration with the Weibull distribution. 
We normalized the duration times, i.e.,
we define the time to write a single FPGA column to be 1 time unit.

For each pair of size and duration values, we shuffled 100 sequences and calculated their
makespan by simulating the processing of a sequence using 
tabu search, greedy, and no defragmentation. More precisely, we successively
place the modules into an array that represents the FPGA. 
If we cannot place a module, because there is no sufficient free space,
either the module has to wait (no defragmentation) or we perform
the tabu search or the greedy strategy to compact the FPGA. 
After the duration time for a module elapsed, it is removed from the array.
Our simulation takes the times needed to place or move a module into 
account; the duration time is prolonged accordingly. 

It turned out that it pays off to use defragmentation for larger modules or larger
duration times. Small modules with small duration time enter and leave the system
so quickly that there is no need for defragmentation, 
see Fig.~\ref{smalltimes-fig}(left) up to an average duration of 50 time units.
At smaller module sizes and execution times, greedy's shorter running
time beats the effectiveness of the tabu search 
(Fig.~\ref{smalltimes-fig}(left) from 75 to 350).
However, as the
average module size (as a fraction of the total area) or execution
length increases, the more compact solution provided by the tabu
search provides a better overall execution time, even with increased
overhead
(Fig.~\ref{smalltimes-fig}(left) from 350 and Fig.~\ref{bigtimes-fig}(left)).
For modules of medium size (compared to the size of the FPGA), the tabu search
decreases the total makespan (Fig.~\ref{smalltimes-fig}(middle) and Fig.~\ref{bigtimes-fig}(middle)).
If the average size of a module approaches or even exceeds half the size of the FPGA, the benefit
of compaction disappears (Fig.~\ref{smalltimes-fig}(right) and Fig.~\ref{bigtimes-fig}(right)).
Note that in this case, compaction is often not even possible because the
modules are too large to be moved.

%% \begin{figure*}[p]
%% \mbox{}\hfill {\epsfig{figure=PLOT3D-3-Makespan-CF0-ShxxxxThxxxxxxx.eps,width=9cm,angle=-90}}\hfill\mbox{}\\
%% \mbox{}\hfill {\epsfig{figure=PLOT3D-5-Makespan-CF0-ShxxxxThxxxxxxx.eps,width=9cm,angle=-90}}\hfill\mbox{}
%% \caption{Comparison of makespans for schedules using tabu search, greedy, and no fragmentation.
%% The average module size range from 10 to 150.
%% The average duration from 1 to 100 (top) and from 500 to 3000 (bottom).}
%% \end{figure*}

\section{Conclusion}
\label{sec:conclusion}

In this paper, we presented a new approach for defragmenting the module
layout on a dynamically reconfigurable device, for example a FPGA, in
a seamless fashion. As the reconfiguration costs continuously decrease
with each new generation of reconfigurable devices and a number of
techniques for task preemption and relocation at a different positions
are conceived (see Koch et al.~\cite{KHT07} for a comparison), task
relocation at runtime becomes a new opportunity for improving the
performance and efficiency of reconfigurable devices. However, this
also poses new challenges, because defragmentation methods developed
so far cannot be applied to reconfigurable devices, as they do not take
into account their special characteristics. For example, many
reconfigurable devices have heterogeneities on their reconfigurable
area, such as memory blocks, DSPs, and CPUs. We presented different
defragmentation strategies to relocate running modules and achieve a
contiguous free space of maximum size.

The presented experiments show in average an increase in the maximal
free space by $30\%$ when applying our defragmentation techniques to 
FPGAs with heterogeneities; on some inputs an increase up to $200\%$
is observed. This additional free space allows earlier execution of
later modules,
so the total execution time is
reduced. This shows that it pays off to prefer a sophisticated heuristic 
for defragmentation (e.g., tabu search) 
over a simple heuristic (i.e., greedy), or over no 
defragmentation at all; provided that the execution times and module sizes 
are not too extreme (i.e., too large or too small compared to the size
of the FPGA).

Obviously, improved algorithmic results can lead to further
improvements. One of the possible extensions considers a more
controlled overall placement of modules, instead of simply fixing
fragmentation. As the necessary algorithmic methods are more involved,
we leave this to future work.

\section*{Acknowledgements}
We like to thank the anonymous referees for many valuable 
suggestions.

\small

% \bibliographystyle{abbrv}
% \bibliographystyle{IEEEtran}
%\bibliographystyle{alpha}

% \bibliography{litfpl}

\newcommand{\noopsort}[1]{}

\iffalse
% \input{appendix}
\newpage

\begin{appendix}
\section{Comparison to \cite{fktvakt-nbddr-08}}
All sections of the paper have been thoroughly revised and updated.
In particular, Section 1.2, Section 2 (notably Theorem 2), and Section 4
are substantially enhanced. 
Section 6, Section 8.2, and Section 8.3 contain new material.
\end{appendix}
\fi

\end{document}